\documentclass[lettersize,journal]{IEEEtran}
\usepackage{amsmath,amsfonts,amssymb,amsthm}
\usepackage[linesnumbered,ruled,vlined]{algorithm2e}
\SetAlgoCaptionSeparator{}
\usepackage{algorithmic}
\usepackage{array}
\usepackage{xcolor}
\usepackage{textcomp}
\usepackage{stfloats}
\usepackage{url}
\usepackage{verbatim}
\usepackage{graphicx}
\usepackage{subcaption}
\usepackage{lipsum}

\usepackage{cite}
\usepackage{CJKutf8}

\usepackage[font=small]{caption}

\DeclareMathOperator{\diag}{diag}
\DeclareMathOperator*{\argmax}{arg\,max}
\DeclareMathOperator*{\argmin}{arg\,min}

\newtheorem{proposition}{Proposition}

\usepackage{hyperref}
\hypersetup{
  colorlinks=true,
  linkcolor=blue,
  filecolor=magenta,
  urlcolor=cyan,
  pdftitle={PLIM-PR-AFDM Paper},
}

\begin{document}

%
\title{\huge{Index Modulation Aided Phase-Rotated AFDM for Efficient Integrated Sensing and Communication: Framework, Optimization and Performance Analysis}}

\author{Shiqi Cui, Zeping Sui, \textit{Member, IEEE}, Tianqi Mao, \textit{Member, IEEE}, Yuanshuo Gang,\\ Guangyao Liu, \textit{Member, IEEE}, Fan Zhang, Miaowen Wen, \textit{Senior Member, IEEE}, Dezhi Zheng, \\  Christos Masouros, \textit{Fellow, IEEE}, and Zhaocheng Wang, \textit{Fellow, IEEE}

  \thanks{This work was supported in part by the Beijing Nova Program under Grant number 202604841146 and National Natural Science Foundation of China under Grant 62401054. \emph{(Corresponding author: Tianqi Mao.)}}
  \thanks{S. Cui, T. Mao, Y. Gang, G. Liu and D. Zheng are with State Key Laboratory of Environment Characteristics and Effects for Near-space, Beijing Institute of Technology, Beijing 100081, China (e-mails: \{cuisqqq, maotq, gangys, liu\_gy, zhengdezhi\}@bit.edu.cn).

    Z. Sui is with the School of Computer Science and Electronic Engineering, University of Essex, CO4 3SQ Colchester, U.K. (e-mail: zepingsui@outlook.com).

    F. Zhang and Z. Wang are with Department of Electronic Engineering, Tsinghua University, Beijing 100084, China. Z. Wang is also with Tsinghua Shenzhen International Graduate School, Shenzhen 518055, China (e-mails: zf22@mails.tsinghua.edu.cn, zcwang@tsinghua.edu.cn).

    M. Wen is with the School of Electronic and Information Engineering, South China University of Technology, Guangzhou 510640, China (e-mail: eemwwen@scut.edu.cn).

    Christos Masouros is with the Department of Electronic and Electrical Engineering, University College London, Torrington Place, London, WC1E 7JE, UK (e-mail: c.masouros@ucl.ac.uk).

}}

\markboth{Journal of \LaTeX\ Class Files,~Vol.~14, No.~8, August~2021}%
{Shell \MakeLowercase{\textit{et al.}}: A Sample Article Using IEEEtran.cls for IEEE Journals}


\IEEEaftertitletext{\vspace{-6mm}}
\maketitle
\begin{abstract}
  Affine frequency division multiplexing (AFDM) has emerged as a competitive waveform for integrated sensing and communication (ISAC) in high-mobility scenarios, but its data-dependent high peak-to-average power ratio (PAPR) and ambiguity function (AF) sidelobes may degrade the power-amplifier efficiency as well as the sensing performance.  To mitigate the issues, one promising candidate is the phase-rotation optimization philosophy, which necessitates auxiliary side information (SI) to guarantee successful data recovery at the receiver. However, using additional signalling overhead to deliver the SI degrades the spectral efficiency. Against the background, this paper proposes a power level index modulation (PLIM)-aided phase-rotated AFDM (PLIM-PR-AFDM) scheme, which embeds the phase-rotation SI within the transmitted block. At the transmitter, the proposed scheme jointly optimizes the waveform by coordinating the PLIM power allocation, phase rotation, and embedding of SI. Specifically, subcarriers are flexibly allocated with different transmit power levels. High-power subcarriers are phase-rotated to suppress the PAPR and fast-slow-time AF sidelobes. By decomposing the fast-slow-time AF sidelobe energy into per-block periodic autocorrelation function (PACF) terms, the sensing objective can be formulated as a PACF-weighted integrated sidelobe level (WISL) minimization problem, while selected low-power subcarriers embed the phase bits as SI through constellation labels. Furthermore, the PLIM power patterns of each transmitted block can convey additional index bits, which compensates for the spectral efficiency loss caused by SI transmission. Moreover, the error performance is analyzed by deriving an upper bound of the average bit error probability (ABEP) and identifying the condition for full diversity. To support practical implementation, a low-complexity yet high-performance receiver, namely residual-aided PLIM approximate message passing, is further developed to jointly recover the index bits, data symbol bits, and phase bits. Simulation results demonstrate that the proposed scheme effectively reduces PAPR, suppresses fast-slow-time AF/PACF sidelobes, and achieves superior error performance compared with conventional solutions.
\end{abstract}

\vspace{-4mm}
\section{Introduction}


Integrated sensing and communication (ISAC) has emerged as a key enabling technology for future wireless networks, where advanced waveforms are expected to support high-rate communications and accurate sensing simultaneously~\cite{ISAC_Survey}. One promising ISAC waveform candidate is affine frequency division multiplexing (AFDM), owing to its robustness against doubly dispersive channels~\cite{AFDM}. Recent studies have explored AFDM from various perspectives, including pilot design~\cite{AFDM_Fan}, pulse-shaping filter design~\cite{AFDM_AF_FANLIU, AFBM}, waveform design~\cite{afdm_isac, AFDM_AF, DPAF_TIT, AFDM_ISAC_TIM, Frame_AFDM_ISAC}, and receiver design~\cite{afdm_joint_est, Multi_path_AFDM, AFDM_FMCW}. Nevertheless, such a promising waveform still faces non-negligible challenges in practical ISAC applications. In particular, the inherent high peak-to-average power ratio (PAPR) degrades the efficiency of nonlinear power amplifiers (PAs)~\cite{PAPR_survey}. Furthermore, the sensing performance of conventional AFDM is usually insufficiently stable, as the sidelobes of the ambiguity function (AF) are influenced by both the chirping parameters and random communication symbols~\cite{CP-OFDM, Sensing_with_commun, Iceberg}. 

To address the PAPR and AF issues, the joint optimization constitutes a real-time waveform design problem within each data block, rather than an offline procedure solely based on long-term signal statistics~\cite{PAPR_survey, Bazzi2023TunablePAPR}. In this context, a growing body of literature has investigated PAPR reduction for AFDM waveforms. Specifically, a grouped pre-chirp selection (GPS) scheme was proposed to vary the pre-chirp parameters across subcarrier groups, enabling the selection of the candidate signal with the lowest PAPR~\cite{AFDM_GPS}. Furthermore, a per-slot scalar chirp-offset selection method was developed to reduce the PAPR while inherently embedding side information (SI) to facilitate low-complexity blind recovery~\cite{choi_AFDM_PAPR}. To mitigate the PAPR without requiring additional SI transmission, optimal pilot positioning and selective pre-chirp-assisted active constellation extension were jointly designed~\cite{PAPR_AFDM_TVT}. More recently, a neural network aided AFDM transceiver optimization framework was proposed to jointly optimize constellation shaping, precoding, and detection, achieving further PAPR reduction through end-to-end waveform learning~\cite{AFDM_NN_PAPR}. These methods can exploit the degrees of freedom (DoF) in the discrete affine Fourier (DAF) domain to reshape the time-domain envelope while preserving desirable communication performance of AFDM. 


In contrast, existing efforts to optimize the AF for AFDM-ISAC primarily focus on statistical characteristics. Specifically, closed-form average-squared AFs were derived for random AFDM signals with and without pulse shaping, showing that the unshaped AF depends on post-chirp parameter and the symbol kurtosis~\cite{AFDM_AF_FANLIU}. Moreover, the continuous-time AFs of AFDM chirp subcarriers and frames were characterized, showing that parameter tuning and guard-symbol insertion support unambiguous and interference-free sensing, respectively~\cite{AFDM_AF}. More generally, a unified statistical framework for the discrete periodic AF and fast-slow-time AF (FST-AF) of arbitrary orthonormal random communication waveforms was proposed in~\cite{DPAF_TIT}, where the AF properties of several communication waveforms were analyzed. However, the aforementioned waveform designs based on the statistical characteristics of signals are not necessarily optimal under instantaneous conditions. Based on this insight, phase-rotated AFDM, which is equivalent to adjusting the pre-chirp parameter, emerges as a promising solution. This approach generally falls into two categories: SI-free and SI-assisted designs. SI-free designs restrict the phase-rotation magnitude to avoid receiver distortion, which inevitably degrades demodulation performance and hinders high-order modulation~\cite{Adaptive_c2}. Alternatively, SI-assisted designs typically assume perfect SI knowledge at the receiver or rely on a finite shared codebook, either reducing spectral efficiency via signaling overhead or limiting the optimization gain due to the restricted codebook size~\cite{AFDM_GPS}. These drawbacks motivate a more efficient SI embedding mechanism.



As a parallel development, index modulation (IM) presents an attractive approach for implicit information embedding. Existing AFDM-IM schemes include subcarrier-activation IM~\cite{AFDM_IM_WCL, AFDM_IM_JSAC}, pre-chirp-domain IM~\cite{AFDM_PIM_TWC, AFDM_MM_IM}, and constellation-mode IM~\cite{AFDM_MM_IM, AFDM_Dual_IM}. However, these methods are often not well matched to phase-rotated AFDM systems, since subcarrier-activation IM typically introduces spurious peaks in the AF~\cite{PLIM}, while pre-chirp-domain and constellation-mode IM degrade the optimization performance of phase-rotated AFDM systems. To address the above-mentioned limitations, we incorporate power level index modulation (PLIM) into the phase-rotated AFDM system~\cite{PLIM, DM_OFDM_CP, Wen_PLIM}. PLIM is a generalized form of subcarrier activation IM that maps index bits onto combinations of subcarrier power levels rather than simply distinguishing between active and inactive subcarriers. Since low-power subcarriers remain active during transmission, PLIM avoids the AF distortions caused by zero-power subcarriers in subcarrier-activation IM schemes, making it suitable for ISAC waveform design~\cite{PLIM}.

Motivated by the aforementioned considerations, this paper presents a PLIM-aided phase-rotated AFDM (PLIM-PR-AFDM) framework for ISAC. Within the proposed framework, each PLIM power pattern carries index bits, while high-power subcarriers provide the DoF for phase rotation to jointly optimize frame-wise PAPR and sensing sidelobes. Specifically, the sensing design targets the weighted integrated sidelobe level (WISL) of the FST-AF over a coherent processing interval. This objective can be further converted into a per-block periodic autocorrelation function (PACF)-WISL minimization problem, which enables independent frame-wise optimization with low complexity. Additionally, selected low-power subcarriers embed the resulting phase bits as internal SI. The receiver can therefore recover the index bits, phase bits, and data symbols from the same transmitted block without an external SI channel. The main contributions of this paper are summarized as follows.
\begin{itemize}
  \item A PLIM-PR-AFDM framework is proposed for ISAC. In each PLIM group, the PLIM power pattern conveys extra index bits. Moreover, high-power subcarriers provide phase-rotation DoF, and selected low-power subcarriers embed the resulting phase bits as internal SI, thereby avoiding an external SI channel or a restrictive finite phase-codebook constraint.
  \item A frame-wise phase-domain waveform optimization framework is developed to jointly suppress PAPR and the sidelobes of the FST-AF, where the sidelobe level associated with FST-AF can be represented by that of the per-block PACF. Subsequently, the dual-function objectives are jointly optimized subject to the PLIM power-pattern constraints. To balance optimization performance and complexity, both a block coordinate descent (BCD) algorithm and a gradient-guided dual-pruned BCD (GDP-BCD) algorithm are proposed.
  \item The error performance of the proposed system is analytically characterized. An upper bound on the average bit error probability (ABEP) is derived from pairwise error events, and a full-diversity condition involving PLIM power patterns, phase rotations, doubly dispersive channels, and PA nonlinearity is established.
  \item A practical residual-aided approximate message passing (RA-AMP) receiver is designed to jointly recover the index bits, data symbol bits, and phase bits. Our RA-AMP combines a minimum mean square error (MMSE) warm start, damped residual pseudo-observations, factorized PLIM group MMSE denoising, and final group MAP detection. Simulation results verify the effectiveness of our proposed scheme in PAPR reduction and FST-AF/PACF sidelobe suppression, with good BER performance.
\end{itemize}

\textit{Notations:} Lower-case, bold lower-case, and bold upper-case letters denote scalars, vectors, and matrices, respectively. $(\cdot)^{T}$, $(\cdot)^{H}$, and $(\cdot)^*$ denote transpose, Hermitian transpose, and conjugation, respectively. $|\mathcal A|$, $\|\mathbf a\|$, and $\diag(\mathbf a)$ denote set cardinality, Euclidean norm, and a diagonal matrix with diagonal $\mathbf a$, respectively. $\mathbb C^{M\times N}$ is the set of $M\times N$ complex matrices; $\mathbf I_N$ and $\mathbf{0}$ denote the $N\times N$ identity matrix and a zero vector of appropriate size. $\ker(\cdot)$ and $\mathbb E\{\cdot\}$ denote null space and expectation. $\mathcal{CN}(\boldsymbol\mu,\mathbf C)$ denotes a complex Gaussian distribution with mean $\boldsymbol\mu$ and covariance $\mathbf C$. $j=\sqrt{-1}$ is the imaginary unit, $(x)_N=x\bmod N$, and $\mathbb Z_N=\{0,1,\ldots,N-1\}$. The Kronecker delta $\delta_{m,n}$ equals $1$ if $m=n$ and $0$ otherwise. For positive $f,g$, $f(t)=\Theta(g(t))$ means $c_1g(t)\le f(t)\le c_2g(t)$ for constants $c_1,c_2>0$ and sufficiently large $t$.

\begin{figure*}[t]
  \centering
  \includegraphics[width=0.86\textwidth]{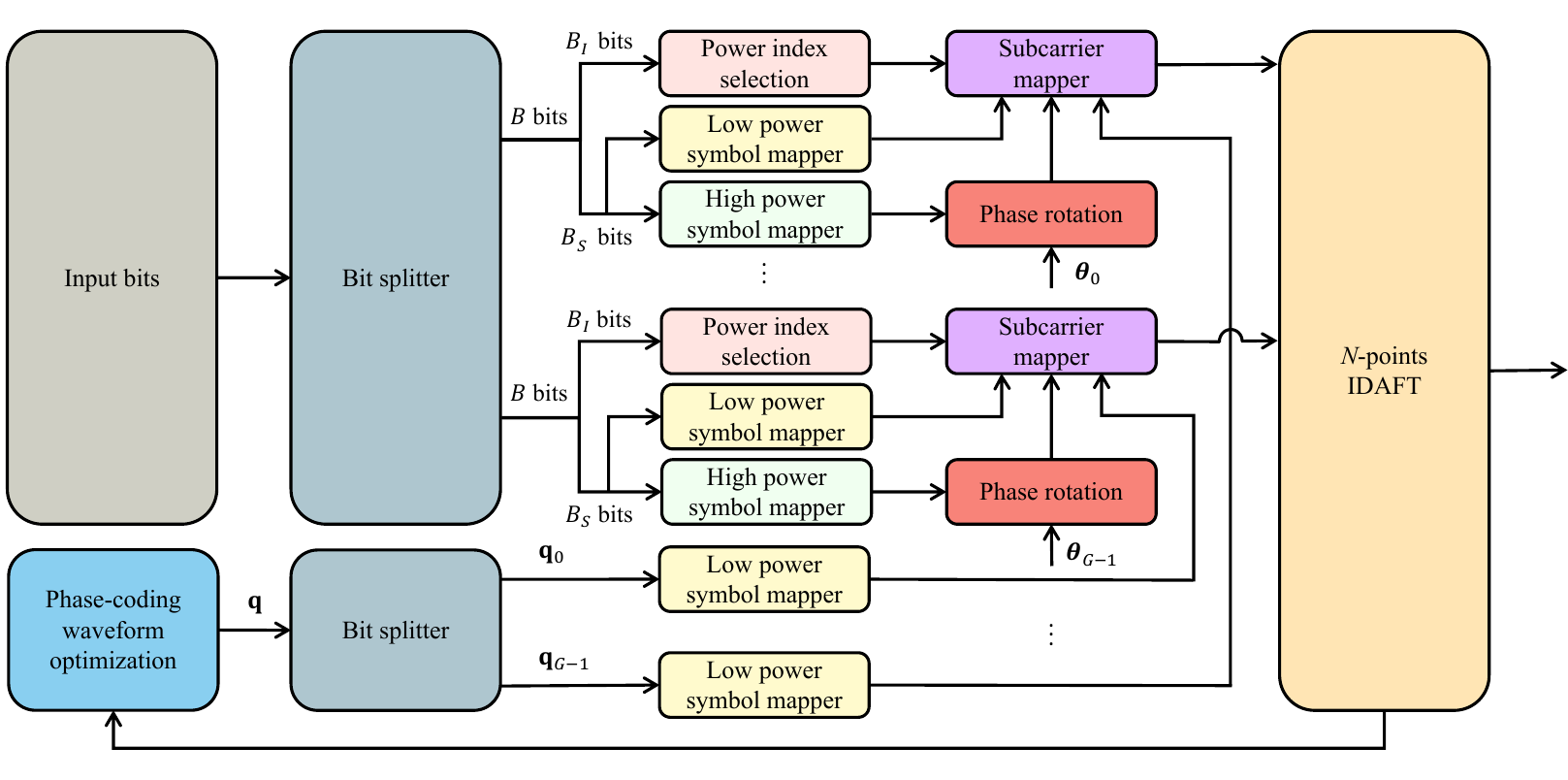}
  \caption{Block diagram of the proposed PLIM-PR-AFDM transmitter.}
  \label{fig:transmitter_model}
  \vspace{-4mm}
\end{figure*}
\vspace{-2mm}
\section{System Model}
\label{sec:system_model}
\subsection{Proposed PLIM-PR-AFDM}
As shown in Fig.~\ref{fig:transmitter_model}, the PLIM-PR-AFDM transmitter performs PLIM power assignment, phase rotation, and in-band SI embedding. The $N$ subcarriers are divided into $G=N/U$ disjoint groups of size $U$, and the subcarrier index set of the $g$-th group is written as
\begin{equation}
  \mathcal G_g=\{gU,gU+1,\ldots,gU+U-1\}, \quad g=0,\ldots,G-1.
\end{equation}
The input bit sequence $\mathbf b$ is evenly partitioned into $G$ groups of $B$ bits. The $g$-th bit group contains $B_{\rm I}$ index bits and $B_{\rm S}=B-B_{\rm I}$ data symbol bits. The index bits select one of $2^{B_{\rm I}}$ usable PLIM power patterns, where $B_{\rm I}=\bigl\lfloor\log_2\binom{U}{K_{\rm H}}\bigr\rfloor$. Let $\ell_g\in\{0,1,\ldots,2^{B_{\rm I}}-1\}$ identify the selected PLIM power pattern in group $g$. This pattern is denoted by $\mathcal P_g(\ell_g)\triangleq(\mathcal H_g,\mathcal L_g)$, where $\mathcal H_g\subseteq\mathcal G_g$ and $\mathcal L_g=\mathcal G_g\setminus\mathcal H_g$ are the high- and low-power subsets with $|\mathcal H_g|=K_{\rm H}$ and $|\mathcal L_g|=K_{\rm L}$, respectively. Thus, $K_{\rm H}$ subcarriers in each group are assigned a high power level, whereas the remaining $K_{\rm L}=U-K_{\rm H}$ subcarriers are assigned a low power level. Collecting the group indices as $\boldsymbol\ell=[\ell_0,\ldots,\ell_{G-1}]^{\rm T}$, the global high- and low-power sets are $\mathcal H=\bigcup_{g=0}^{G-1}\mathcal H_g$ and $\mathcal L=\bigcup_{g=0}^{G-1}\mathcal L_g$. By setting the high-to-low power ratio to $\rho=\alpha_{\rm H}/\alpha_{\rm L}>1$ and normalizing the average power per subcarrier to unity, the coefficients of high and low power levels, denoted as $\alpha_{\rm L}$ and $\alpha_{\rm H}$, can be derived as
\begin{equation}
  \alpha_{\rm L}=\frac{U}{K_{\rm H}\rho+K_{\rm L}},
  \qquad
  \alpha_{\rm H}=\rho\alpha_{\rm L}.
\end{equation}

Since high-power subcarriers make a larger weighted contribution to the waveform properties, phase rotation is applied exclusively to them. Specifically, write $\mathcal H_g=\{m_{g,0}^{\rm H},\ldots,m_{g,K_{\rm H}-1}^{\rm H}\}$, where $m_{g,k}^{\rm H}$ denotes the index of the $k$-th high-power subcarrier in group $g$. Each such subcarrier is associated with one binary phase bit $q_{g,k}\in\{0,1\}$, based on which it undergoes the phase rotation
\begin{equation}
  \theta_{g,k}=(2q_{g,k}-1)\theta_{\rm p},
  \quad k=0,\ldots,K_{\rm H}-1,
\end{equation}
where $\theta_{\rm p}$ is the phase-rotation magnitude. We define the phase-bit sequence and the corresponding high-power phase-rotation vector for group $g$ as $\mathbf q_g\triangleq[q_{g,0},\ldots,q_{g,K_{\rm H}-1}]^{\rm T}$ and $\boldsymbol\theta_g\triangleq[\theta_{g,0},\ldots,\theta_{g,K_{\rm H}-1}]^{\rm T}$, respectively. The block-level phase-bit vector is $\mathbf q\triangleq[\mathbf q_0^{\rm T},\ldots,\mathbf q_{G-1}^{\rm T}]^{\rm T}$.

For SI transmission, the phase-bit vector $\mathbf q$ is embedded in-band rather than conveyed through a separate channel. Specifically, dual constellation alphabets $\mathcal X_{\rm H}$ and $\mathcal X_{\rm L}$ are employed for the high- and low-power subcarriers, respectively. For $i\in\{\mathrm H,\mathrm L\}$, $\mathcal X_i$ has cardinality $M_i$ and bit depth $B_i=\log_2M_i$. In each group $g$, the $K_{\rm H}$ phase bits in $\mathbf q_g$ are first packed into $R_{\rm SI}=\lceil K_{\rm H}/B_{\rm L}\rceil$ predetermined low-power subcarriers. This construction uses a total of $N_{\rm SI}=GR_{\rm SI}$ low-power SI symbols per AFDM block.\footnote{To embed all phase bits within each group, the system parameters are selected such that $R_{\rm SI}\le K_{\rm L}$.} Apart from the $K_{\rm H}$ constellation-label positions occupied by phase bits, all high- and low-power constellation labels convey data symbol bits. Hence, the number of information bits per group is
\begin{equation}
  \label{eq:user_bits_per_group}
  B=B_{\rm I}+B_{\rm S}
  =B_{\rm I}+K_{\rm H}B_{\rm H}+K_{\rm L}B_{\rm L}-K_{\rm H},
  \quad \text{bits/group}.
\end{equation}


\begin{table}[!t]
  \renewcommand{\arraystretch}{1.3}
  \captionsetup{font=footnotesize,labelfont=normalfont,textfont=sc,
    labelsep=newline,justification=centering,singlelinecheck=false}
  \caption{Mapping from the Group Index $\ell_g$ to the In-Group PLIM Power Pattern for $U=4$ and $K_{\rm H}=2$}
  \label{tab:plim_patterns}
  \centering
  \footnotesize
  \begin{tabular}{|c|c|c|c|}
    \hline
    $\ell_g$ & $\mathcal H_g$ & $\mathcal L_g$ & Power pattern\\
    \hline
    $0$ & $\{gU,\,gU+1\}$ & $\{gU+2,\,gU+3\}$ & (H, H, L, L) \\
    \hline
    $1$ & $\{gU,\,gU+2\}$ & $\{gU+1,\,gU+3\}$ & (H, L, H, L) \\
    \hline
    $2$ & $\{gU,\,gU+3\}$ & $\{gU+1,\,gU+2\}$ & (H, L, L, H) \\
    \hline
    $3$ & $\{gU+1,\,gU+2\}$ & $\{gU,\,gU+3\}$ & (L, H, H, L) \\
    \hline
  \end{tabular}
\end{table}

\noindent
\textbf{Example:} Assume $U=4$, $K_{\rm H}=2$, $K_{\rm L}=2$, $M_{\rm H}=16$, and $M_{\rm L}=4$. Since $B_{\rm I}=\bigl\lfloor\log_2 \binom{U}{K_{\rm H}}\bigr\rfloor=2$, $2^{B_{\rm I}}=4$ power patterns are used. Ordering the candidate high-power index sets lexicographically and retaining the first $2^{B_I}$ sets yields the one-to-one mapping in Table~\ref{tab:plim_patterns}. 
Each group consists of two high-power subcarriers for 16 quadrature amplitude modulation (QAM) data-symbol transmission and phase rotation, together with two low-power quadrature phase-shift keying (QPSK) subcarriers, one of which conveys 2-bit SI, while the other conveys 2-bit data symbols. Including the $B_{\rm I}=2$ index bits, \eqref{eq:user_bits_per_group} gives $B=12$ information bits per group.

Having specified the PLIM power pattern, phase rotation, and SI embedding, we now construct the time-domain signal. Let the mapped constellation vector be $\mathbf a=[a_0,\ldots,a_{N-1}]^{\rm T}$, where $a_m\in\mathcal X_{\rm H}/\mathcal X_{\rm L}$ for $m\in\mathcal H/\mathcal L$. The subcarrier-dependent power coefficient and phase rotation are compactly defined as
\begin{equation}
  \bigl(\alpha_m(\boldsymbol\ell),\theta_m(\mathbf q)\bigr)=
  \begin{cases}
    \bigl(\alpha_{\rm H},(2q_{g,k}-1)\theta_{\rm p}\bigr),
      & m\in\mathcal H,\\
    \bigl(\alpha_{\rm L},0\bigr),
      & m\in\mathcal L.
  \end{cases}
\end{equation}
Conditioned on $\mathbf a$, whose argument is suppressed for brevity, the $m$-th DAFT-domain symbol is
\begin{equation}
  x_m(\boldsymbol\ell,\mathbf q)=
  \sqrt{\alpha_m(\boldsymbol\ell)}\,a_m e^{j\theta_m(\mathbf q)},
  \quad m=0,\ldots,N-1.
\end{equation}
By collecting these symbols, the DAFT-domain vector is $\mathbf x(\boldsymbol\ell,\mathbf q)=[x_0(\boldsymbol\ell,\mathbf q),\ldots,x_{N-1}(\boldsymbol\ell,\mathbf q)]^{\rm T}$. The transmitted AFDM block is then given by
\begin{equation}
  \label{eq:tx_time_block}
  \mathbf s(\boldsymbol\ell,\mathbf q)=\mathbf A^{H}\mathbf x(\boldsymbol\ell,\mathbf q).
\end{equation}
where $\mathbf A^{H}=\boldsymbol\Lambda_{c_1}^{H}\mathbf F^{H}\boldsymbol\Lambda_{c_2}^{H}$, $[\mathbf F]_{m,n}=N^{-1/2}e^{-j2\pi mn/N}$, and $\boldsymbol\Lambda_c=\diag(e^{-j2\pi c n^2})_{n=0}^{N-1}$. Here, $c_1$ and $c_2$ are the post-chirp and pre-chirp parameters, respectively. Suppressing the dependence on $(\boldsymbol\ell,\mathbf q)$ for brevity, the $n$-th sample of \eqref{eq:tx_time_block} is
\begin{equation}
  \begin{split}
    s[n]&=\frac{1}{\sqrt N}e^{j2\pi c_1n^2}
    \sum_{m=0}^{N-1}\sqrt{\alpha_m(\boldsymbol\ell)}a_m \\
    &\quad\times e^{j\theta_m(\mathbf q)}e^{j2\pi c_2m^2}e^{j2\pi mn/N}.
  \end{split}
\end{equation}

\subsection{Channel Model and Receiver}
The proposed PLIM-PR-AFDM waveform is emitted for simultaneous communication and sensing. The corresponding channel models and receiver processing are described separately below.

\subsubsection{Communication Branch}
After insertion of the chirp-periodic prefix (CPP), the modulated signal propagates through a $P$-path doubly dispersive channel with impulse response
\begin{equation}
  h(\tau,\nu)=
  \sum_{p=1}^{P}
  h_p
  \delta(\tau-\tau_p)
  \delta(\nu-\nu_p),
\end{equation}
where $h_p\in\mathbb C$, $\tau_p\in\mathbb Z_N$, and $\nu_p=f_{{\rm D},p}/\Delta f$ are the complex gain, discrete delay index, and Doppler shift normalized by the subcarrier spacing $\Delta f$, respectively. The channel-gain vector $\mathbf h=[h_1,\ldots,h_P]^{\rm T}$ follows
\begingroup
\setlength{\abovedisplayskip}{1pt}
\setlength{\belowdisplayskip}{1pt}
\setlength{\abovedisplayshortskip}{0pt}
\setlength{\belowdisplayshortskip}{1pt}
\begin{equation}
  \label{eq:rayleigh_fading}
  \mathbf h\sim\mathcal{CN}\left(\mathbf{0},\frac{1}{P}\mathbf I_P\right).
\end{equation}
\endgroup
According to~\cite{AFDM}, full diversity is achieved by selecting the post-chirp parameters as
\begin{equation}
  \label{eq:c1_condition}
  c_1 = \frac{2\left( \nu_{\rm max} + k_\nu \right) + 1}{2N},
\end{equation}
and $c_2\in\mathbb R\setminus\mathbb Q$ being irrational, where $k_\nu\in\mathbb Z_{\geq 0}$ is a guard integer for fractional-Doppler leakage, selected such that
\begin{equation}
  2\left(\nu_{\rm max}+k_\nu\right)\left(\tau_{\rm max}+1\right)
  +\tau_{\rm max}<N.
  \label{eq:afdm_guard_condition}
\end{equation}
with $\nu_{\rm max}$ and $\tau_{\rm max}$ being the maximum Doppler and delay indices, respectively.

At the receiver, after CPP removal, the received time-domain sample vector $\mathbf r=[r[0],\ldots,r[N-1]]^{\rm T}$ is given by
\begingroup
\setlength{\abovedisplayskip}{1pt}
\setlength{\belowdisplayskip}{1pt}
\setlength{\abovedisplayshortskip}{0pt}
\setlength{\belowdisplayshortskip}{1pt}
\begin{equation}
  \mathbf r=
  \sum_{p=1}^{P}
  h_p
  \mathbf H_{t,p}
  \mathbf s+
  \mathbf n_t,
\end{equation}
\endgroup
where $\mathbf H_{t,p}=\boldsymbol\Gamma_{{\rm CPP},p}\boldsymbol\Delta_{\nu_p}\boldsymbol\Pi^{\tau_p}$. The matrices $\boldsymbol\Pi^{\tau_p}$, $\boldsymbol\Delta_{\nu_p}=\diag(e^{j2\pi\nu_p n/N})_{n=0}^{N-1}$, and $\boldsymbol\Gamma_{{\rm CPP},p}$ model cyclic delay, Doppler modulation, and the CPP-induced phase, respectively, while $\mathbf n_t$ is the additive white Gaussian noise (AWGN) vector, following $\mathcal{CN}(\mathbf 0,N_0\mathbf I_N)$. Subsequently, applying the DAFT gives
\begingroup
\setlength{\abovedisplayskip}{1pt}
\setlength{\belowdisplayskip}{1pt}
\setlength{\abovedisplayshortskip}{0pt}
\setlength{\belowdisplayshortskip}{1pt}
\begin{equation}
  \mathbf y=
  \mathbf A\mathbf r=
  \sum_{p=1}^{P}
  h_p
  \mathbf A
  \mathbf H_{t,p}
  \mathbf A^{H}
  \mathbf x+
  \mathbf w,
\end{equation}
\endgroup
where $\mathbf w\triangleq\mathbf A\mathbf n_t$. Define the DAFT-domain channel matrix of the $p$-th path as $\mathbf H_p\triangleq\mathbf A\mathbf H_{t,p}\mathbf A^{H}$ and the effective channel as $\mathbf H_{\rm eff}\triangleq\sum_{p=1}^{P}h_p\mathbf H_p$. The input-output relation can be expressed as
\begingroup
\setlength{\abovedisplayskip}{1pt}
\setlength{\belowdisplayskip}{1pt}
\setlength{\abovedisplayshortskip}{0pt}
\setlength{\belowdisplayshortskip}{1pt}
\begin{equation}
  \label{eq:comm_daft_model}
  \mathbf y=
  \mathbf H_{\rm eff}
  \mathbf x+
  \mathbf w.
\end{equation}
\endgroup

Based on \eqref{eq:comm_daft_model}, the PLIM power-pattern index vector $\boldsymbol\ell$, embedded phase-bit vector $\mathbf q$, and data-symbol vector $\mathbf a$ can be jointly recovered from $\mathbf y$ by maximum-likelihood (ML) detection. Its search complexity, however, grows excessively with the constellation orders and the number of candidate PLIM power patterns. To mitigate this complexity, Section~\ref{sec:receiver_design} develops the RA-AMP detector for practical implementation.

\subsubsection{Sensing Branch}
For sensing, a monostatic ISAC receiver coherently processes $M_{\rm s}$ consecutive AFDM blocks with matched filtering for target detection. Denote the $u$-th transmitted block by $\mathbf s_u$. The waveform over one coherent processing interval (CPI) is stacked as
$\mathbf s_{\rm CPI}
\triangleq
[
\mathbf s_0^{\rm T},\cdots,\mathbf s_{M_{\rm s}-1}^{\rm T}
]^{\rm T}
\in\mathbb C^{NM_{\rm s}}$.
Under the fast-slow-time (FST) approximation~\cite{DPAF_TIT}, the Doppler phase shifts within each block is neglected, whereas those across successive blocks are retained. For target $i=1,\ldots,Q$, the reflection coefficient, delay-bin index, and slow-time Doppler-bin index are denoted by $\beta_i\in\mathbb C$, $d_i\in\mathbb Z_N$, and $\iota_i\in\mathbb Z_{M_{\rm s}}$, respectively. After CPP removal, the sensing echo is given by
\begingroup
\setlength{\abovedisplayskip}{1pt}
\setlength{\belowdisplayskip}{1pt}
\setlength{\abovedisplayshortskip}{0pt}
\setlength{\belowdisplayshortskip}{1pt}
\begin{equation}
  \mathbf r_{\rm sen}
  =
  \sum_{i=1}^{Q}
  \beta_i
  \left(
    \mathbf D_{M_{\rm s},\iota_i}
    \otimes
    \mathbf J_{N,d_i}
  \right)
  \mathbf s_{\rm CPI}
  +
  \mathbf n_{\rm sen},
  \label{eq:fst_echo_model}
\end{equation}
\endgroup
where $\mathbf n_{\rm sen}\sim\mathcal{CN}(\mathbf 0,\sigma_{\rm sen}^{2}\mathbf I_{NM_{\rm s}})$ is sensing noise. The inter-block Doppler and periodic delay operators in \eqref{eq:fst_echo_model} are, respectively,
$\mathbf D_{M_{\rm s},\iota}
\triangleq
\diag\!\left(e^{j2\pi \iota u/M_{\rm s}}\right)_{u=0}^{M_{\rm s}-1}$,
$\iota\in\mathbb Z_{M_{\rm s}}$, and
$[\mathbf J_{N,k}\mathbf v]_n
\triangleq v[(n-k)_N]$,
$k,n\in\mathbb Z_N$,
for any $\mathbf v\in\mathbb C^N$.

With matched-filter processing, the FST-AF characterizes the delay-Doppler sensing performance of the transmitted CPI waveform and is defined as~\cite{DPAF_TIT}
\begin{equation}
  \mathcal{A}_{\rm FST}(k,\iota)
  \triangleq
  \mathbf s_{\rm CPI}^H
  \left(
    \mathbf D_{M_{\rm s},\iota}^{*}
    \otimes
    \mathbf J_{N,k}^{\rm T}
  \right)
  \mathbf s_{\rm CPI}.
  \label{eq:fst_af_def}
\end{equation}
To connect this sensing model to block-wise waveform design, define the PACF of block $u$ as $r_{u}[k]\triangleq\mathbf s_u^H\mathbf J_{N,k}^{\rm T}\mathbf s_u$, $k\in\mathbb Z_N$. From the block-diagonal structure of $\mathbf D_{M_{\rm s},\iota}^{*}\otimes\mathbf J_{N,k}^{\rm T}$, we have
\begin{align}
  \mathcal{A}_{\rm FST}(k,\iota)
  &=
  \sum_{u=0}^{M_{\rm s}-1}
  r_{u}[k]
  e^{-j2\pi \iota u/M_{\rm s}} .
  \label{eq:fst_as_dft_pacf}
\end{align}
According to Parseval's Theorem, the sidelobe energy of FST-AF over $\mathcal K\triangleq\mathbb Z_N\setminus\{0\}$ satisfies
\begin{equation}
  \sum_{k\in\mathcal K}
  \eta_k
  \sum_{\iota=0}^{M_{\rm s}-1}
  |\mathcal{A}_{\rm FST}(k,\iota)|^2
  =
  M_{\rm s}
  \sum_{u=0}^{M_{\rm s}-1}
  \sum_{k\in\mathcal K}
  \eta_k
  |r_u[k]|^2 ,
  \label{eq:weighted_fst_pacf_relation}
\end{equation}
where each $k\in\mathcal K$ is assigned a weight $\eta_k\geq0$. The right-hand side is $M_{\rm s}$ times the sum of the per-block weighted integrated sidelobe level of the PACF. Therefore, minimizing the PACF-WISL of each block independently reduces the total Doppler-integrated FST-AF sidelobe energy over the CPI, without jointly optimizing all $M_{\rm s}$ blocks.

\section{Proposed ISAC Waveform Optimization}
This section jointly optimizes the PAPR and the FST-AF sidelobe level within the proposed PLIM-PR-AFDM framework. Based on \eqref{eq:weighted_fst_pacf_relation}, the FST-AF sidelobe level over a CPI can be expressed as a linear combination of the per-block PACF sidelobe energies. Therefore, the original mission can be transferred as a finite-alphabet optimization problem aiming at jointly minimizing the PAPR and PACF-WISL for each AFDM block.
\vspace{-6mm}
\subsection{Problem Formulation}
Consider one arbitrary AFDM block (block index omitted for brevity), where $\mathbf \ell$ and the transmitted data symbols are fixed. In the proposed approach, the optimization variable is the block-level phase-bit vector $\mathbf q=[\mathbf q_0^{\rm T},\ldots,\mathbf q_{G-1}^{\rm T}]^{\rm T}$ defined in Section~\ref{sec:system_model}. The group-level phase-bit vector $\mathbf q_g$ belongs to the finite alphabet
\begin{equation}
  \mathbf q_g\in\mathcal Q\triangleq\{0,1\}^{K_{\rm H}},
  \qquad |\mathcal Q|=N_{\rm q}=2^{K_{\rm H}}.
  \label{eq:group_phase_bit_alphabet}
\end{equation}
Each $\mathbf q_g$ determines the rotation angles of the corresponding high-power subcarriers, i.e., $\boldsymbol\theta_g=(2\mathbf q_g-\mathbf 1)\theta_{\rm p}$, and its entries are embedded in the low-power subcarriers as SI. For brevity, we denote $\mathbf x(\mathbf q)\triangleq\mathbf x(\boldsymbol\ell,\mathbf q)$ and $\mathbf s(\mathbf q)\triangleq\mathbf A^H\mathbf x(\mathbf q)$ as the transmitted DAFT- and time-domain waveforms, as a function of the associated SI. On the one hand, the PAPR of the AFDM block is defined as
\begin{equation}
  J_{\rm PAPR}(\mathbf q)
  =
  \frac{
    N\|\mathbf s(\mathbf q)\|_{\infty}^{2}
  }{
    \|\mathbf s(\mathbf q)\|_{2}^{2}
  }.
  \label{eq:papr_obj}
\end{equation}
On the other hand, according to \eqref{eq:weighted_fst_pacf_relation}, the PACF-WISL of the AFDM block can be written as
\begin{equation}
  J_{\rm WISL}(\mathbf q)
  =
  \sum_{k\in\mathcal{K}}
  \eta_k
  \left|r^{(\mathbf q)}[k]\right|^2.
  \label{eq:pacf_wisl_obj}
\end{equation}
Here, $r^{(\mathbf q)}[k]$ denotes the per-block PACF sample $r_u[k]$ in \eqref{eq:weighted_fst_pacf_relation}, evaluated at $\mathbf s(\mathbf q)$ with the block index $u$ suppressed.

To balance the two metrics on a common scale, a normalized joint objective is established as
\begin{equation}
  J(\mathbf q)
  =
  \lambda
  \frac{
    J_{\rm PAPR}(\mathbf q)
  }{
    J_{\rm PAPR}^{\rm ref}
  }
  +
  (1-\lambda)
  \frac{
    J_{\rm WISL}(\mathbf q)
  }{
    J_{\rm WISL}^{\rm ref}
  },
  \label{eq:joint_papr_pacf_obj}
\end{equation}
where $\lambda\in[0,1]$ controls the communication-sensing tradeoff. Besides, the constants $J_{\rm PAPR}^{\rm ref}$ and $J_{\rm WISL}^{\rm ref}$ denote the metrics of the conventional unoptimized AFDM waveform for identical data symbols and fixed side information which are determined by $\mathbf q$ for each block.

Finally, the resulting finite-alphabet optimization problem is formulated as
\begin{subequations}\label{eq:problem}
  \begin{align}
    \min_{\mathbf q}\quad
    &J(\mathbf q) \label{eq:opt_problem_obj}\\
    \mathrm{s.t.}\quad
    &\mathbf q_g\in\mathcal Q,\quad g=0,1,\ldots,G-1, \label{eq:opt_problem_state}\\
    &\boldsymbol\theta_g=(2\mathbf q_g-\mathbf 1)\theta_{\rm p},\quad g=0,\ldots,G-1, \label{eq:opt_problem_phase_high}\\
    &\theta_m=0,\quad m\in\mathcal L. \label{eq:opt_problem_phase_low}
  \end{align}
\end{subequations}
The corresponding search space has cardinality $|\mathcal Q|^G=2^{GK_{\rm H}}=2^{|\mathcal H|}$. Hence, exhaustive enumeration is computationally prohibitive for large $N$ or large group phase-pattern alphabets, motivating the low-complexity BCD-based algorithm proposed below.
\subsection{GDP-BCD Optimization}
Conventional BCD reduces the search space by updating a single group of variables during each iteration~\cite{Hong2016BSUM}. Explicitly, for the current $\mathbf q$, it replaces $\mathbf q_g$ with a candidate $\widetilde{\mathbf q}\in\mathcal Q$ and forms $\mathbf q_{g\leftarrow\widetilde{\mathbf q}}=[\mathbf q_0^{\rm T},\ldots,\mathbf q_{g-1}^{\rm T},\widetilde{\mathbf q}^{\rm T},\mathbf q_{g+1}^{\rm T},\ldots,\mathbf q_{G-1}^{\rm T}]^{\rm T}$. The corresponding increments are $\Delta\mathbf x_{g,\widetilde{\mathbf q}}=\mathbf x(\mathbf q_{g\leftarrow\widetilde{\mathbf q}})-\mathbf x(\mathbf q)$ and $\Delta\mathbf s_{g,\widetilde{\mathbf q}}=\mathbf A^{H}\Delta\mathbf x_{g,\widetilde{\mathbf q}}$. BCD evaluates each candidate $\widetilde{\mathbf q}\in\mathcal Q$ and accepts the best one only if the objective strictly decreases. Although this monotone local search is more efficient than exhaustive enumeration, its computational complexity is still dominated by exact evaluations over all candidate phase-bit patterns. To further lower the complexity, we propose the GDP-BCD algorithm.

Specifically, our GDP-BCD algorithm incorporates a two-stage pruning mechanism for the promising groups and candidate phase-bit patterns prior to the exact evaluation. To prioritize these groups and patterns, we formulate a surrogate criterion based on a first-order Taylor approximation of the objective function to replace the computationally expensive exact evaluations. Since the argument $\mathbf q$ of $J(\mathbf q)$ is discrete, the gradient of $J$ is evaluated with respect to the continuous waveform $\mathbf s$. Given a small perturbation $\Delta \mathbf{s}$, the objective value is approximated as
\begin{equation}
  J(\mathbf s+\Delta\mathbf s)
  \approx J(\mathbf s)+2\operatorname{Re}\{\mathbf g_{\mathbf s}^{H}\Delta\mathbf s\},
\end{equation}
where
\begin{equation}
  \mathbf g_{\mathbf s}
  =\lambda\frac{\nabla_{\mathbf s^\ast}J_{\rm PAPR}}{J_{\rm PAPR}^{\rm ref}}
  +(1-\lambda)\frac{\nabla_{\mathbf s^\ast}J_{\rm WISL}}{J_{\rm WISL}^{\rm ref}}
\end{equation}
denote the Wirtinger gradient of the real-valued objective function. Assuming that the peak time-domain sample of $\mathbf s$ is unique, i.e., $|s_{n^\star}|^2>|s_n|^2$ for all $n\neq n^\star$, the gradients of the PAPR and PACF-WISL terms can be expressed as
\begin{align}
  \nabla_{\mathbf s^\ast}J_{\rm PAPR}
  &=\frac{N}{\|\mathbf s\|_2^2}s_{n^\star}\mathbf e_{n^\star}
  -\frac{N|s_{n^\star}|^2}{\|\mathbf s\|_2^4}\mathbf s, \\
  \nabla_{\mathbf s^\ast}J_{\rm WISL}
  &=\sum_{k=1}^{N-1}\eta_k
  \left(r^\ast[k]\mathbf J_{N,k}^{\rm T}\mathbf s+r[k](\mathbf J_{N,k}^{\rm T})^{H}\mathbf s\right),
  \label{eq:wisl_grad}
\end{align}
where $n^\star=\argmax_n |s_n|^2$ for $n=0,1,...,N-1$, $\mathbf e_{n^\star}\in\mathbb C^N$ denotes the basis vector satisfying $[\mathbf e_{n^\star}]_n=\delta_{n,n^\star}$, and $r[k]$ represents the per-block PACF sample $r_u[k]$ in \eqref{eq:weighted_fst_pacf_relation} (index $u$ omitted). Both the PACF sequence $\{r[k]\}$ and the circularly shifted products $\mathbf J_{N,k}^{\rm T}\mathbf s$ and $(\mathbf J_{N,k}^{\rm T})^{H}\mathbf s$ in \eqref{eq:wisl_grad} can be calculated using FFT/IFFT operations, thereby attaining efficient computation of the WISL gradient.

By defining the DAFT-domain gradient as $\mathbf g_{\mathbf x}=\mathbf A\mathbf g_{\mathbf s}$, the predicted decrease resulting from the replacement of $\mathbf q_g$ with $\widetilde{\mathbf q}$ is subsequently defined as 
\begin{equation}
D_{g,\widetilde{\mathbf q}}=J(\mathbf q)-J(\mathbf q_{g\leftarrow\widetilde{\mathbf q}})\approx -2\operatorname{Re}\{\mathbf g_{\mathbf x}^{H}\Delta\mathbf x_{g,\widetilde{\mathbf q}}\}.
\end{equation}
A larger value of $D_{g,\widetilde{\mathbf q}}$ implies that replacing $\mathbf q_g$ with $\widetilde{\mathbf q}$ is more likely to reduce the exact objective value.
For both group and candidate-pattern pruning, a table of predicted decreases $\mathbf D=[D_{g,\widetilde{\mathbf q}}]_{g,\widetilde{\mathbf q}\in\mathcal Q}\in\mathbb R^{G\times N_{\rm q}}$ is established, where $D_{g,\mathbf q_g}=-\infty$ is assigned to the current pattern. We then prune groups and candidate phase-bit patterns according to the predicted decreases. Specifically, the first pruning stage retains the $K_G$ groups with the highest scores, given by
\begin{equation}
  \mathcal G_{\rm GDP}
  =\argmax_{\mathcal G'\subseteq\mathbb Z_G,\,|\mathcal G'|=K_G}
  \sum_{g\in\mathcal G'}S_g,
  \
  S_g=\max_{\widetilde{\mathbf q}\in\mathcal Q\setminus\{\mathbf q_g\}}D_{g,\widetilde{\mathbf q}}.
  \label{eq:gdp_group_pruning}
\end{equation}
For each retained group, the second pruning stage preserves only $K_C$ candidate phase-bit patterns with the highest scores, yielding
\begin{equation}
  \widehat{\mathcal Q}_g
  =\argmax_{\mathcal Q'\subseteq\mathcal Q\setminus\{\mathbf q_g\},\,|\mathcal Q'|=K_C}
  \sum_{\widetilde{\mathbf q}\in\mathcal Q'}D_{g,\widetilde{\mathbf q}},
  \quad g\in\mathcal G_{\rm GDP},
  \label{eq:gdp_candidate_pruning}
\end{equation}
where $K_C\leq N_{\rm q}-1$. The objective function is then evaluated exclusively for $\widetilde{\mathbf q}\in\widehat{\mathcal Q}_g$ using the exact expression in \eqref{eq:joint_papr_pacf_obj}. Although the proposed algorithm employs a first-order approximation of the objective function during the pruning phase to reduce complexity, the exact evaluation of the pruned candidate set is mandatory for the convergence of the optimization process.
The detailed procedure of the proposed GDP-BCD algorithm is summarized in \textbf{Algorithm~\ref{alg:gdp_bcd_papr_pacf}}.

\begin{algorithm}[t]
  \caption{GDP-BCD for Waveform Optimization}
  \label{alg:gdp_bcd_papr_pacf}
  \footnotesize
  \KwIn{Initial phase-bit vector $\mathbf q^{(0)}$, fixed PLIM power-pattern vector $\boldsymbol\ell$ and payload data symbols, maximum number of iterations $T_{\rm GDP}$, group budget $K_G$, candidate-pattern budget $K_C$, tolerance $\varepsilon_{\rm GDP}$.}
  \KwOut{Optimized phase-bit vector $\mathbf q$.}
  $\mathbf q\leftarrow\mathbf q^{(0)}$, $J_{\rm cur}\leftarrow J(\mathbf q)$\;
  \For{$t=1,2,\ldots,T_{\rm GDP}$}{
    $\mathrm{improved}\leftarrow\mathrm{false}$\;
    Construct $\mathbf x=\mathbf x(\mathbf q)$ and $\mathbf s=\mathbf A^{H}\mathbf x$\;
    Compute $\mathbf g_{\mathbf s}=\lambda\nabla_{\mathbf s^\ast}J_{\rm PAPR}/J_{\rm PAPR}^{\rm ref}+(1-\lambda)\nabla_{\mathbf s^\ast}J_{\rm WISL}/J_{\rm WISL}^{\rm ref}$\;
    Compute $\mathbf g_{\mathbf x}=\mathbf A\mathbf g_{\mathbf s}$\;
    \For{$g=0,1,\ldots,G-1$}{
      \ForEach{$\widetilde{\mathbf q}\in\mathcal Q\setminus\{\mathbf q_g\}$}{
        Form sparse $\Delta\mathbf x_{g,\widetilde{\mathbf q}}$ and compute $D_{g,\widetilde{\mathbf q}}=-2\operatorname{Re}\{\mathbf g_{\mathbf x}^{H}\Delta\mathbf x_{g,\widetilde{\mathbf q}}\}$\;
      }
      $S_g\leftarrow\max_{\widetilde{\mathbf q}\neq\mathbf q_g}D_{g,\widetilde{\mathbf q}}$\;
    }
    Obtain $\mathcal G_{\rm GDP}$ using \eqref{eq:gdp_group_pruning}\;
    \ForEach{$g\in\mathcal G_{\rm GDP}$}{
      Obtain $\widehat{\mathcal Q}_g$ using \eqref{eq:gdp_candidate_pruning}\;
      \ForEach{$\widetilde{\mathbf q}\in\widehat{\mathcal Q}_g$}{
        Compute $J_{g,\widetilde{\mathbf q}}^{\rm exact}=J(\mathbf q_{g\leftarrow\widetilde{\mathbf q}})$\;
      }
      $\mathbf q_g^\star\leftarrow\argmin_{\widetilde{\mathbf q}\in\widehat{\mathcal Q}_g}J_{g,\widetilde{\mathbf q}}^{\rm exact}$\;
      \If{$J_{\rm cur}-J_{g,\mathbf q_g^\star}^{\rm exact}>\varepsilon_{\rm GDP}$}{
        $\mathbf q_g\leftarrow\mathbf q_g^\star$, $J_{\rm cur}\leftarrow J_{g,\mathbf q_g^\star}^{\rm exact}$\;
        $\mathrm{improved}\leftarrow\mathrm{true}$\;
      }
    }
    \If{$\mathrm{improved}=\mathrm{false}$}{Stop the iteration\;}
  }
\end{algorithm}

\subsection{Algorithm Complexity}
With FFT-based AFDM modulation and PACF evaluation, the computational complexity of one exact objective evaluation $C_{\rm ex}$ is $\mathcal O(N\log N)$. Since BCD tests $(N_{\rm q}-1)$ alternatives for each group in every iteration, its complexity after at most $T_{\rm BCD}$ iterations is $C_{\rm BCD}=\mathcal O\left(T_{\rm BCD}G(N_{\rm q}-1)C_{\rm ex}\right)$.

In each iteration of GDP-BCD, it first computes the waveform and DAFT-domain gradients with complexity $C_{\rm grad}=\mathcal O(N\log N)$ and then scores all group-candidate-pattern pairs through sparse products $\mathbf g_{\mathbf x}^{H}\Delta\mathbf x_{g,\widetilde{\mathbf q}}$. If each candidate update changes $\zeta$ DAFT-domain entries, the scoring complexity is $\mathcal O\left(G(N_{\rm q}-1)\zeta\right)$. Only $K_G$ selected groups and $K_C$ candidate phase-bit patterns per selected group are then exactly verified, giving
\begin{equation}
  \!\!C_{\rm GDP}\!=\!\mathcal O\left(T_{\rm GDP}\left(C_{\rm grad} \!+ \! G(N_{\rm q}-1)\zeta \!+ \! K_GK_CC_{\rm ex}\right)\right).
\end{equation}
Thus, GDP-BCD reduces the dominant exact evaluations from $G(N_{\rm q}-1)$ to $K_GK_C$ per iteration. When exact evaluations dominate and the scoring overhead is small, ${C_{\rm BCD}^{\rm iter}}/{C_{\rm GDP}^{\rm iter}}$ approaches $G(N_{\rm q}-1)/(K_GK_C)$. Hence, $K_G\ll G$ and $K_C\ll N_{\rm q}-1$ yield substantial complexity reduction while retaining exact verification for accepted updates.

\section{Error Performance Analysis}
This section characterizes the error performance of the proposed PLIM-PR-AFDM block under perfect channel state information. Since practical PAPR-oriented operation is directly related to PA nonlinearity, we derive the pairwise error probability (PEP), BER upper-bound, and the full-diversity condition when the transmitted waveform is distorted by a nonlinear PA. The ideal-PA result is then obtained as a simplified special case.

\subsection{Nonlinear-PA Mismatched Analysis}
Let $B_{\rm u}=GB$ denote the number of information bits per AFDM block, where $B$ is defined in \eqref{eq:user_bits_per_group}. Each label $\mathbf b\in\{0,1\}^{B_{\rm u}}$ determines $\boldsymbol\ell(\mathbf b)$ and a data-symbol realization. Conditioned on that fixed pair, let $\mathcal Q_{\rm opt}(\mathbf b)=\argmin_{\mathbf q\in\mathcal Q^G}J(\mathbf q)$ denote the set of optimal phase-bit vectors, and choose its lexicographically smallest member $\mathbf q^\star(\mathbf b)$. This fixed rule defines the deterministic codebook
\begin{equation}
  \mathcal C
  =\left\{
    \mathbf x\big(\boldsymbol\ell(\mathbf b),\mathbf q^\star(\mathbf b)\big):
    \mathbf b\in\{0,1\}^{B_{\rm u}}
  \right\},
  \
  |\mathcal C|=2^{B_{\rm u}}.
  \label{eq:transmit_codebook}
\end{equation}
These assumptions make the deterministic encoding map injective and hence justify the cardinality in \eqref{eq:transmit_codebook}. Index the labels and codewords as $\mathbf b_i$ and $\mathbf x_i=\mathbf x\big(\boldsymbol\ell(\mathbf b_i),\mathbf q^\star(\mathbf b_i)\big)$, respectively, and let the corresponding time-domain AFDM block be $\mathbf s_i=\mathbf A^{H}\mathbf x_i$. For a nonlinear PA, the sample-wise output can be written as $\bar s_i[n]=f_{\rm PA}(s_i[n])$, where $f_{\rm PA}(\cdot)$ is the nonlinear transfer function of the PA. Thus, the actual DAFT-domain vector with PA distortion is
\begin{equation}
  \label{eq:pa_daft_codeword}
  \bar{\mathbf x}_i
  =
  \mathbf A\bar{\mathbf s}_i
  =
  \mathbf A f_{\rm PA}(\mathbf A^{H}\mathbf x_i).
\end{equation}

Define the ideal detection matrix and the actual PA-distorted transmit matrix as $\mathbf S_i \triangleq\left[\mathbf H_1\mathbf x_i, \mathbf H_2\mathbf x_i, \ldots, \mathbf H_P\mathbf x_i \right]$ and $\bar{\mathbf S}_i\triangleq\left[\mathbf H_1\bar{\mathbf x}_i, \mathbf H_2\bar{\mathbf x}_i, \ldots, \mathbf H_P\bar{\mathbf x}_i \right]$, respectively. When the $i$-th codeword is transmitted, the received signal is given by
\begin{equation}
  \label{eq:pa_actual_received_signal}
  \mathbf y
  =
  \bar{\mathbf S}_i\mathbf h+
  \mathbf w,
  \quad
  \mathbf w\sim\mathcal{CN}(\mathbf 0,N_0\mathbf I_N).
\end{equation}
Since the receiver is unaware of the PA nonlinearity, the ML detector based on the codebook can be formulated as
\begin{equation}
  \label{eq:pa_mismatched_ml_detector}
  \hat i
  =
  \argmin_{j}
  \left\|
  \mathbf y-
  \mathbf S_j\mathbf h
  \right\|_2^2 .
\end{equation}
For the pairwise error event $(i\rightarrow j)$, define
\begin{align}
  \boldsymbol\Psi_{ij}
  &\triangleq
  \mathbf S_i-
  \mathbf S_j
  =
  \left[
    \mathbf H_1\boldsymbol\delta_{ij},
    \ldots,
    \mathbf H_P\boldsymbol\delta_{ij}
  \right], \label{eq:Psi_ij}\\
  \mathbf E_i
  &\triangleq
  \bar{\mathbf S}_i-
  \mathbf S_i
  =
  \left[
    \mathbf H_1(\bar{\mathbf x}_i-\mathbf x_i),
    \ldots,
    \mathbf H_P(\bar{\mathbf x}_i-\mathbf x_i)
  \right],
\end{align}
where $\boldsymbol\delta_{ij}\triangleq\mathbf x_i-\mathbf x_j$. Given the channel realization $\mathbf h$, under the pairwise error event $(i\rightarrow j)$ we have $\|\boldsymbol\Psi_{ij}\mathbf h+\mathbf E_i\mathbf h+\mathbf w\|_2^2\le\|\mathbf E_i\mathbf h+\mathbf w\|_2^2$, which can be reformulated as
\begin{equation}
  2\operatorname{Re}\{(\boldsymbol\Psi_{ij}\mathbf h)^{H}\mathbf w\} \!
  \le \!
  -\|\boldsymbol\Psi_{ij}\mathbf h\|_2^2
  -2\operatorname{Re}\{(\boldsymbol\Psi_{ij}\mathbf h)^{H}\mathbf E_i\mathbf h\}. \!
  \label{eq:pa_mismatch_error_event}
\end{equation}
Since $2\operatorname{Re}\{(\boldsymbol\Psi_{ij}\mathbf h)^{H}\mathbf w\}\sim\mathcal N(0,2N_0\|\boldsymbol\Psi_{ij}\mathbf h\|_2^2)$, the conditional mismatched PEP can be expressed as
\begin{equation}
  \label{eq:pa_mismatch_conditional_pep}
  P_{i\rightarrow j}^{\rm PA}(\mathbf h)
  =
  Q\left(
    \frac{
      \mathbf h^{H}\mathbf M_{ij}^{(i)}\mathbf h
    }{
      \sqrt{2N_0\mathbf h^{H}\boldsymbol\Omega_{ij}\mathbf h}
    }
  \right),
\end{equation}
where $Q(\cdot)$ denotes the Gaussian $Q$-function, defined by $Q(x)\triangleq\frac{1}{\sqrt{2\pi}}\int_x^{\infty}e^{-t^2/2}\,dt.$ Moreover, $\boldsymbol\Omega_{ij} \triangleq\boldsymbol\Psi_{ij}^{H}\boldsymbol\Psi_{ij}$ and $\mathbf M_{ij}^{(i)}\triangleq\boldsymbol\Omega_{ij}+\boldsymbol\Psi_{ij}^{H}\mathbf E_i+\mathbf E_i^{H}\boldsymbol\Psi_{ij}$. The matrix $\mathbf M_{ij}^{(i)}$ is Hermitian but not necessarily positive semidefinite, and it depends on the transmitted codeword $i$. Hence, the error events $(i\rightarrow j)$ and $(j\rightarrow i)$ generally have different PEPs. Upon averaging \eqref{eq:pa_mismatch_conditional_pep} over the fading vector, we have $P_{i\rightarrow j}^{\rm PA}=\mathbb E_{\mathbf h}\left[P_{i\rightarrow j}^{\rm PA}(\mathbf h)\right]$. When $\mathbf h\sim\mathcal{CN}(\mathbf 0,\mathbf R_h)$, $P_{i\rightarrow j}^{\rm PA}$ can be evaluated as
\begin{align}
  \label{eq:pa_mismatch_average_pep_integral}
  P_{i\rightarrow j}^{\rm PA}\!
  &=\!
  \int_{\mathbb C^P}\!
  Q\!\left(
    \frac{
      \mathbf h^{H}\mathbf M_{ij}^{(i)}\mathbf h
    }{
      \sqrt{2N_0\mathbf h^{H}\boldsymbol\Omega_{ij}\mathbf h}
    }
  \right)\!
  \frac{
    \exp(-\mathbf h^{H}\mathbf R_h^{-1}\mathbf h)
  }{
    \pi^P\det(\mathbf R_h)
  }
  d\mathbf h, \\
  &\overset{\eqref{eq:rayleigh_fading}}{=}
  \frac{P^P}{\pi^P}
  \int_{\mathbb C^P}
  Q\left(
    \frac{
      \mathbf h^{H}\mathbf M_{ij}^{(i)}\mathbf h
    }{
      \sqrt{2N_0\mathbf h^{H}\boldsymbol\Omega_{ij}\mathbf h}
    }
  \right)
  e^{-P\|\mathbf h\|_2^2}
  d\mathbf h .
\end{align}

If all codewords are transmitted with equal probability, the BER upper-bound under nonlinear PA is
\begin{equation}
  \label{eq:pa_mismatch_abep_union_bound}
  P_b^{\rm PA}
  \le
  \frac{1}{B_{\rm u}|\mathcal C|}
  \sum_{i=1}^{|\mathcal C|}
  \sum_{\substack{j=1\\j\neq i}}^{|\mathcal C|}
  d_{\rm H}(\mathbf b_i,\mathbf b_j)
  P_{i\rightarrow j}^{\rm PA},
\end{equation}
where the Hamming distance between the two $B_{\rm u}$-bit labels is $d_{\rm H}(\mathbf b_i,\mathbf b_j)\triangleq\sum_{n=1}^{B_{\rm u}}\mathbb I\!\left([\mathbf b_i]_n\neq[\mathbf b_j]_n\right).$ The following proposition gives the necessary and sufficient condition under which nonlinear PA mismatch preserves the full multipath diversity.

\begin{proposition}
  \label{prop:pa_mismatch_full_diversity}
  Let $\mathbf h\sim\mathcal{CN}(\mathbf0,\mathbf R_h)$ with fixed $\mathbf R_h\succ\mathbf0$. For the fixed injective codebook in \eqref{eq:transmit_codebook} and fixed path matrices, assume $\boldsymbol\Psi_{ij}\ne\mathbf0$ for all $i\ne j$ and a deterministic, measurable, memoryless PA with finite outputs on all codebook samples. Take $\gamma=1/N_0\to\infty$ at unit average input power, with PA parameters, input-back-off (IBO), and gain normalization fixed. Then the ML detector achieves full $P$-path diversity, $P_b^{\rm PA}=\Theta(\gamma^{-P})$, if and only if
  \begin{equation}
    \label{eq:pa_mismatch_full_diversity_condition}
    \mathbf M_{ij}^{(i)}\succ\mathbf0,
    \quad \forall i\ne j.
  \end{equation}
\end{proposition}
\begin{proof}
  A competing metric smaller than that of the transmitted codeword forces at least one bit error. Hence, for every $i\ne j$,
  \begin{equation}
    \label{eq:pa_mismatch_ber_lower_bound}
    P_b^{\rm PA}\ge
    \frac{P_{i\rightarrow j}^{\rm PA}}{B_{\rm u}|\mathcal C|}.
  \end{equation}
  Together with \eqref{eq:pa_mismatch_abep_union_bound}, this reduces the claim to pairwise decay rates. For a fixed pair, write $\mathbf M=\mathbf M_{ij}^{(i)}$ and $\boldsymbol\Omega=\boldsymbol\Omega_{ij}$, both in $\mathbb C^{P\times P}$. If $\mathbf M\succ\mathbf0$, then $\boldsymbol\Omega\succ\mathbf0$, since $\mathbf v\in\ker(\boldsymbol\Omega)$ implies $\mathbf v^H\mathbf M\mathbf v=0$. Thus, \eqref{eq:pa_mismatch_conditional_pep} lies between $Q(b\sqrt\gamma\|\mathbf h\|_2)$ and $Q(a\sqrt\gamma\|\mathbf h\|_2)$ for fixed $0<a\le b$. Gaussian averaging gives $P_{i\rightarrow j}^{\rm PA}=\Theta(\gamma^{-P})$, proving sufficiency.

  Conversely, a negative direction of $\mathbf M$ yields a nonempty open set of negative margins with positive fading probability, so $P_{i\rightarrow j}^{\rm PA}\to\Pr(\mathbf h^H\mathbf M\mathbf h<0)>0$. If instead $\mathbf M\succeq\mathbf0$ is singular, then $\ker(\boldsymbol\Omega)\subseteq\ker(\mathbf M)$, implying $\mathbf M\preceq c\boldsymbol\Omega$ for some fixed $c>0$. Consequently,
  \begin{equation}
    \label{eq:pa_mismatch_singular_pep_lower_bound}
    \begin{aligned}
      P_{i\rightarrow j}^{\rm PA}
      &\ge\mathbb E_{\mathbf h}\!\left[
        Q\!\left(\sqrt{\frac{\gamma c}{2}\mathbf h^H\mathbf M\mathbf h}\right)
      \right]\\
      &\ge c_0\gamma^{-r},
    \end{aligned}
  \end{equation}
  where $r=\operatorname{rank}(\mathbf M)<P$ and $c_0>0$ is independent of $\gamma$. The last bound holds for sufficiently large $\gamma$ since the Gaussian quadratic form has $r$ nonzero eigenvalues, including the constant $Q(0)=1/2$ case when $r=0$. Therefore, \eqref{eq:pa_mismatch_ber_lower_bound} precludes diversity $P$, establishing necessity.
\end{proof}

\subsection{Ideal-PA Case}
The ideal-PA case follows from the preceding mismatched analysis by setting $f_{\rm PA}(s)=s$. Then $\bar{\mathbf x}_i=\mathbf x_i$, $\mathbf E_i=\mathbf 0$, and $\mathbf M_{ij}^{(i)}=\boldsymbol\Omega_{ij}.$ The conditional PEP is the simplified form of \eqref{eq:pa_mismatch_conditional_pep}, given by
\begin{equation}
  \label{eq:abep_conditional_pep}
  P(\mathbf x_i\rightarrow\mathbf x_j\mid\mathbf h)
  =
  Q\left(
    \sqrt{
    \frac{\mathbf h^{H}\boldsymbol\Omega_{ij}\mathbf h}{2N_0}}
  \right).
\end{equation}
Since the numerator is a nonnegative quadratic form, Craig's representation of the $Q$-function holds 
\begin{equation}
  Q(x)
  =
  \frac{1}{\pi}
  \int_0^{\pi/2}
  \exp\!\left(
    -\frac{x^2}{2\sin^2\vartheta}
  \right)d\vartheta,
  \quad x\ge0.
  \label{eq:craig_q_representation}
\end{equation}
Substituting the quadratic-form identity $\mathbb E[e^{-s\mathbf h^{H}\boldsymbol\Omega_{ij}\mathbf h}]=\det(\mathbf I_P+s\mathbf R_h\boldsymbol\Omega_{ij})^{-1}$ yields the PEP as
\begin{align}
  \label{eq:abep_average_pep_general}
  P(\mathbf x_i\rightarrow\mathbf x_j)
  &=
  \frac{1}{\pi}\int_{0}^{\pi/2}
  \det\left(
    \mathbf I_P+
    \frac{\mathbf R_h\boldsymbol\Omega_{ij}}
    {4N_0\sin^2\vartheta}
  \right)^{-1}d\vartheta , \\
  & \overset{\eqref{eq:rayleigh_fading}}{=}
  \frac{1}{\pi}\!\!\int_{0}^{\pi/2} \!\!
  \det \! \left(
    \mathbf I_P+
    \frac{\boldsymbol\Omega_{ij}}
    {4PN_0\sin^2\vartheta}
  \right)^{-1} \!\!d\vartheta .
\end{align}
Let $\{\lambda_{ij,1},\ldots,\lambda_{ij,r_{ij}}\}$ be the nonzero eigenvalues of $\boldsymbol\Omega_{ij}$, where $r_{ij}=\operatorname{rank}(\boldsymbol\Omega_{ij})=\operatorname{rank}(\boldsymbol\Psi_{ij})$. Then
\begin{equation}
  \label{eq:abep_average_pep_eigen}
  P(\mathbf x_i\rightarrow\mathbf x_j)
  =
  \frac{1}{\pi}\!\int_{0}^{\pi/2}\!
  \prod_{k=1}^{r_{ij}}
  \left(
    1+
    \frac{\lambda_{ij,k}}
    {4PN_0\sin^2\vartheta}
  \right)^{-1}\!\!d\vartheta .
\end{equation}
From \eqref{eq:abep_average_pep_eigen}, using the Wallis integral, the high-SNR PEP is
\begin{equation}
  \label{eq:abep_high_snr_pep}
  P(\mathbf x_i\rightarrow\mathbf x_j)
  \overset{\gamma\rightarrow\infty}{\sim}
  \frac{
    P^{r_{ij}}\binom{2r_{ij}}{r_{ij}}
  }{
    2\prod_{k=1}^{r_{ij}}\lambda_{ij,k}
  }
  \gamma^{-r_{ij}}.
\end{equation}

Under equiprobable codewords, the BER upper-bound using ideal PA can be expressed as
\begin{equation}
  \label{eq:abep_union_bound}
  P_b
  \le
  \frac{1}{B_{\rm u}|\mathcal C|}
  \sum_{i=1}^{|\mathcal C|}
  \sum_{\substack{j=1\\j\neq i}}^{|\mathcal C|}
  d_{\rm H}(\mathbf b_i,\mathbf b_j)
  P(\mathbf x_i\rightarrow\mathbf x_j).
\end{equation}
Therefore, the diversity order of each codeword pair is governed by $r_{ij}=\operatorname{rank}(\boldsymbol\Psi_{ij})$, while the product $\prod_{k=1}^{r_{ij}}\lambda_{ij,k}$ determines the coding gain. The ideal-PA full-diversity condition is the specialization of Proposition~\ref{prop:pa_mismatch_full_diversity} with $\mathbf E_i=\mathbf 0$, namely
\begin{equation}
  \label{eq:full_diversity_condition}
  \min_{\boldsymbol\delta\in\mathcal D_{\rm u}}
  \operatorname{rank}
  \left[
    \mathbf H_1\boldsymbol\delta,
    \ldots,
    \mathbf H_P\boldsymbol\delta
  \right]
  = P,
\end{equation}
where $\mathcal D_{\rm u}=\{\boldsymbol\delta_{ij}:d_{\rm H}(\mathbf b_i,\mathbf b_j)>0\}$. This condition requires both channel separability and codeword separability. According to~\cite{AFDM}, when \eqref{eq:c1_condition} is satisfied, different path components satisfy the channel-separability condition. For the codewords, avoiding $\boldsymbol\delta_{ij}=\mathbf 0$ for any distinct pair requires
\begin{align}
  \label{eq:codeword_separability_conditions}
  \sqrt{\alpha_{\rm H}}e^{j\theta}\mathcal X_{\rm H}
  \cap
  \sqrt{\alpha_{\rm L}}\mathcal X_{\rm L}
  &= \varnothing,
  \quad \forall\theta\in\{-\theta_{\rm p},\theta_{\rm p}\}, \\
  e^{j\theta_{\rm p}}\mathcal X_{\rm H}
  \cap
  e^{-j\theta_{\rm p}}\mathcal X_{\rm H}
  &= \varnothing.
\end{align}
When \eqref{eq:full_diversity_condition} and \eqref{eq:codeword_separability_conditions} hold, the system achieves full $P$-path diversity under ideal PA.

\section{Proposed RA-AMP Receiver}
\label{sec:receiver_design}

This section develops a practical receiver for the proposed PLIM-PR-AFDM system. An RA-AMP receiver is proposed for the DAFT-domain received vector $\mathbf y$ in \eqref{eq:comm_daft_model}. It converts the doubly dispersive DAFT-domain observation into iteratively refined pseudo-observations and then performs factorized group-wise denoising by exploiting the local SI partition.

\subsection{RA-AMP Receiver}
Our RA-AMP first generates a damped matched-filter pseudo-observation and then performs factorized group maximum a posteriori (MAP) detection from that pseudo-observation. The DAFT-domain pseudo-observation used by the group denoiser is modeled as
\begin{equation}
  \label{eq:rapamp_scalar_model}
  \mathbf z_{\rm use}^{(t)}=\mathbf x+\boldsymbol\xi^{(t)},
  \quad \boldsymbol\xi^{(t)} \ \dot{\sim} \ \mathcal{CN}(\mathbf 0,v_{\rm lin}^{(t)}\mathbf I_N).
\end{equation}
where $v_{\rm lin}^{(t)}$ denotes the effective noise variance of the pseudo-observation, which will be defined in \eqref{eq:rapamp_variance_clip}. RA-AMP uses an MMSE warm start:
\begin{equation}
  \label{eq:rapamp_lmmse_filter}
  \mathbf W=(\mathbf H_{\rm eff}^{H}\mathbf H_{\rm eff}+N_0\mathbf I_N)^{-1}\mathbf H_{\rm eff}^{H},
  \quad \mathbf z_{\rm raw}=\mathbf W\mathbf y,
\end{equation}
and $\mathbf G_{\rm eq}=\mathbf W\mathbf H_{\rm eff}$. Let $d_m=[\mathbf G_{\rm eq}]_{m,m}$. The diagonal-normalized observation $z_{{\rm Eq},m}$ and its variance $\sigma_m^2$ estimate are
\begin{equation}
  \label{eq:rapamp_lmmse_unbias}
  z_{{\rm Eq},m}=z_{{\rm raw},m}/d_m,
\end{equation}
\begin{equation}
  \label{eq:rapamp_lmmse_variance}
  \sigma_m^2=\frac{N_0\sum_n|W_{m,n}|^2+
  \sum_{j\ne m}|[\mathbf G_{\rm eq}]_{m,j}|^2}{|d_m|^2}.
\end{equation}
The initial pseudo-noise variance is $v^{(0)}=N^{-1}\sum_{m=0}^{N-1}\sigma_m^2$, and the initial estimate obtained by PLIM group-MMSE denoising is given by
\begin{equation}
  \label{eq:rapamp_initial_denoising}
  \widehat{\mathbf x}^{(0)}=\mathcal D_{\rm PLIM}^{\rm F}(\mathbf z_{\rm Eq},v^{(0)}),
\end{equation}
where $\mathcal D_{\rm PLIM}^{\rm F}(\cdot, \cdot)$ denotes the factorized PLIM group-MMSE denoiser whose compact expression will be given in \eqref{eq:rapamp_plim_denoiser_expression}.

During iteration $t$, let $\bar h^2=\|\mathbf H_{\rm eff}\|_F^2/N$ and $\mu=\eta_{\rm RA}/\bar h^2$, where $\eta_{\rm RA}$ controls the matched-filter step size. RA-AMP first computes the residual error as
\begin{equation}
  \label{eq:rapamp_residual}
  \mathbf r^{(t)}=\mathbf y-\mathbf H_{\rm eff}\widehat{\mathbf x}^{(t)},
\end{equation}
and forms the matched-filter pseudo-observation
\begin{equation}
  \label{eq:rapamp_linear_observation}
  \mathbf z_{\rm lin}^{(t)}=
  \widehat{\mathbf x}^{(t)}+
  \mu\mathbf H_{\rm eff}^{H}\mathbf r^{(t)}.
\end{equation}
This update is one residual-gradient correction of the objective $\|\mathbf y-\mathbf H_{\rm eff}\mathbf x\|^2$. The scalar variance used by the nonlinear denoiser is estimated from the residual energy, yielding
\begin{equation}
  \label{eq:rapamp_variance_raw}
  \widetilde v_{\rm lin}^{(t)}=N_0\mu^2\bar h^2+
  c_r\frac{\max\{\|\mathbf r^{(t)}\|^2/N-N_0,0\}}{\bar h^2},
\end{equation}
where $c_r$ denotes the residual-variance calibration factor. To avoid overconfident early iterations, this variance is clipped as
\begin{equation}
  \label{eq:rapamp_variance_clip}
  v_{\rm lin}^{(t)}=
  \min\left\{
    \max\left\{\widetilde v_{\rm lin}^{(t)},c_r N_0\right\},
    10\max\left\{v^{(0)},N_0\right\}
  \right\}.
\end{equation}

With $\mathbf z_{\rm use}^{(-1)}$ initialized by $\mathbf z_{\rm Eq}$, RA-AMP stabilizes the residual correction by damping both the pseudo-observation and the denoiser output, hence we have
\begin{align}
  \label{eq:rapamp_memory_damping}
  \mathbf z_{\rm use}^{(t)}
  &=\lambda_d\mathbf z_{\rm lin}^{(t)}+(1-\lambda_d)\mathbf z_{\rm use}^{(t-1)}, \\
  \label{eq:x_den_t}
  \mathbf x_{\rm den}^{(t)}
  &=\mathcal D_{\rm PLIM}^{\rm F}(\mathbf z_{\rm use}^{(t)},v_{\rm lin}^{(t)}), \\
  \widehat{\mathbf x}^{(t+1)}
  &=\lambda_d\mathbf x_{\rm den}^{(t)}+(1-\lambda_d)\widehat{\mathbf x}^{(t)},
\end{align}
where $\mathbf z_{\rm use}^{(t)}$ is the damped pseudo-observation, $\mathbf x_{\rm den}^{(t)}$ is the denoiser output, and $\lambda_d\in(0,1]$ is the damping factor.

We next specify the nonlinear module $\mathcal D_{\rm PLIM}^{\rm F}(\cdot, \cdot)$ in \eqref{eq:x_den_t}. For a candidate PLIM power pattern $\ell$ in group $g$, the high-power, SI, and ordinary low-power index sets are respectively expressed as
\begin{subequations}
  \label{eq:rapamp_candidate_sets}
  \begin{align}
    \mathcal H_{g,\ell}
    &=\{m_{g,\ell,0}^{\rm H},m_{g,\ell,1}^{\rm H},\ldots,m_{g,\ell,K_{\rm H}-1}^{\rm H}\}, \\
    \mathcal S_{g,\ell}
    &=\{m_{g,\ell,0}^{\rm SI},m_{g,\ell,1}^{\rm SI},\ldots,m_{g,\ell,R_{\rm SI}-1}^{\rm SI}\}, \\
    \mathcal D_{g,\ell}
    &=\mathcal G_g\setminus(\mathcal H_{g,\ell}\cup\mathcal S_{g,\ell}),
  \end{align}
\end{subequations}
where $R_{\rm SI}=|\mathcal S_{g,\ell}|$ denotes the number of SI subcarriers per group. The $r$-th SI subcarrier embeds the phase-bit subset $\mathcal I_r\subseteq\{1,\ldots,K_{\rm H}\}$, where $\mathcal I_r\cap\mathcal I_{r'}=\varnothing$ for $r\neq r'$ and $\bigcup_{r=1}^{R_{\rm SI}}\mathcal I_r=\{1,\ldots,K_{\rm H}\}$.
Given an input pair $(\mathbf z,v)$, the denoiser output is the posterior average of the conditional local means over the PLIM power patterns, i.e.,
\begin{equation}
  \label{eq:rapamp_plim_denoiser_expression}
  \left[\mathcal D_{\rm PLIM}^{\rm F}(\mathbf z,v)\right]_m
  =
  \sum_{\ell\in\mathcal B_{\rm use}}
  \rho_g(\ell;\mathbf z_g,v)\bar\mu_m(\ell;\mathbf z_g,v),
\end{equation}
where $\mathcal B_{\rm use}$ denotes the available PLIM power-pattern set, $\rho_g(\ell;\mathbf z_g,v)$ is the posterior probability of the PLIM power pattern, and $\bar\mu_m(\ell;\mathbf z_g,v)$ is the conditional local posterior mean with $\mathbf z_g\triangleq\{z_m:m\in\mathcal G_g\}$. Their evaluation proceeds in the following three steps: local likelihood evaluation, block-wise factorization, and posterior-moment aggregation.

\emph{1) Local likelihood evaluation:} Given $\mathbf z_{\rm use}^{(t)}$ and $v_{\rm lin}^{(t)}$, let $\mathsf p_i(x)=\exp\left(-|z_i-x|^2/v_{\rm lin}^{(t)}\right)$ denote the Gaussian observation kernel, where $z_i$ is the $i$-th entry of $\mathbf z_{\rm use}^{(t)}$. The local likelihoods for high-power, SI, and ordinary low-power subcarriers are expressed as
\begin{subequations}
\label{eq:rapamp_local_likelihoods}
\begin{align}
  \psi_{\rm H}(m,q)
  &=\frac{1}{M_{\rm H}}\sum_{a\in\mathcal X_{\rm H}}
  \mathsf p_m(x_{\rm H}(a,q)), \\
  \psi_{\rm SI}(m,\boldsymbol\tau)
  &=\mathsf p_m(x_{\rm SI}(\boldsymbol\tau)), \\
  \psi_{\rm L}(d)
  &=\frac{1}{M_{\rm L}}\sum_{a\in\mathcal X_{\rm L}}
  \mathsf p_d(x_{\rm L}(a)).
\end{align}
\end{subequations}
where $x_{\rm H}(a,q)=\sqrt{\alpha_{\rm H}}e^{j(2q-1)\theta_{\rm p}}a$, $x_{\rm SI}(\boldsymbol\tau)=\sqrt{\alpha_{\rm L}}\chi_{\rm L}(\boldsymbol\tau)$, and $x_{\rm L}(a)=\sqrt{\alpha_{\rm L}}a$ are the noiseless local symbols. The local SI label vector $\boldsymbol\tau\in\{0,1\}^{|\mathcal I_r|}$ has $2^{|\mathcal I_r|}$ possible values, with $|\mathcal I_r|\leq B_{\rm L}$.

\emph{2) Block-wise factorization:} For a fixed pattern $\ell$, each SI subcarrier couples only the phase bits in its own subset $\mathcal I_r$. Hence the corresponding local block factor is
\begin{equation}
\label{eq:rapamp_local_factor}
\varphi_{g,\ell,r}^{(t)}(\boldsymbol\tau)
=
\psi_{\rm SI}(m_{g,\ell,r}^{\rm SI},\boldsymbol\tau)
\prod_{k\in\mathcal I_r}\psi_{\rm H}(m_{g,\ell,k}^H,\tau_k),
\end{equation}
for $\boldsymbol\tau\in\{0,1\}^{|\mathcal I_r|}$. Then, the unnormalized likelihood of the complete group state $(\ell,\mathbf q)$ is factorized as
\begingroup
\setlength{\abovedisplayskip}{1pt}
\setlength{\belowdisplayskip}{1pt}
\setlength{\abovedisplayshortskip}{0pt}
\setlength{\belowdisplayshortskip}{1pt}
\begin{equation}
\label{eq:rapamp_factorized_weight}
\begin{aligned}
W_g^{(t)}(\ell,\mathbf q)
&=
\prod_{d\in\mathcal D_{g,\ell}}\psi_{\rm L}(d)
\prod_{r=1}^{R_{\rm SI}}
\varphi_{g,\ell,r}^{(t)}(\mathbf q_{\mathcal I_r}),
\end{aligned}
\end{equation}
\endgroup
where $\mathbf q_{\mathcal I_r}$ is the phase-bit block embedded in the $r$-th SI symbol. Since the subsets $\{\mathcal I_r\}_{r=1}^{R_{\rm SI}}$ are non-overlapping, marginalizing over $\mathbf q$ reduces to independent local sums, yielding
\begingroup
\setlength{\abovedisplayskip}{1pt}
\setlength{\belowdisplayskip}{1pt}
\setlength{\abovedisplayshortskip}{0pt}
\setlength{\belowdisplayshortskip}{1pt}
\setlength{\jot}{0pt}
\begin{align}
\label{eq:rapamp_local_partition}
Z_{g,\ell,r}^{(t)}
&=
\sum_{\boldsymbol\tau\in\{0,1\}^{|\mathcal I_r|}}
\varphi_{g,\ell,r}^{(t)}(\boldsymbol\tau), \\
Z_g^{(t)}(\ell)
&=
p(\ell)
\prod_{d\in\mathcal D_{g,\ell}}\psi_{\rm L}(d)
\prod_{r=1}^{R_{\rm SI}}Z_{g,\ell,r}^{(t)},
\end{align}
\endgroup
where $p(\ell)$ can be omitted under a uniform prior over the PLIM power patterns. The posterior probability of the PLIM power pattern and the conditional local phase posterior are respectively expressed as
\begingroup
\setlength{\abovedisplayskip}{1pt}
\setlength{\belowdisplayskip}{1pt}
\setlength{\abovedisplayshortskip}{0pt}
\setlength{\belowdisplayshortskip}{1pt}
\begin{equation}
\label{eq:rapamp_pattern_posterior}
\rho_g^{(t)}(\ell)
=
\frac{Z_g^{(t)}(\ell)}{
\sum_{\ell'\in\mathcal B_{\rm use}} Z_g^{(t)}(\ell')}, \quad \pi_{g,\ell,r}^{(t)}(\boldsymbol\tau)
=
\frac{\varphi_{g,\ell,r}^{(t)}(\boldsymbol\tau)}{Z_{g,\ell,r}^{(t)}}.
\end{equation}
\endgroup

\emph{3) Posterior-moment aggregation:} The posterior moments are computed from the local posteriors. For high-power symbols, let us define $\mu_{\rm H}(m,q)$ and $\nu_{\rm H}(m,q)$ as the posterior mean and second-order moment of the high-power symbol at subcarrier $m$ conditioned on the phase bit $q$, respectively, i.e.,
\begin{align}
\label{eq:rapamp_high_symbol_moments}
\mu_{\rm H}(m,q)
&=
\frac{\sum_{a\in\mathcal X_{\rm H}}x_{\rm H}(a,q)\omega_{\rm H}(m,q,a)}
{M_{\rm H}\psi_{\rm H}(m,q)}, \\
\nu_{\rm H}(m,q)
&=
\frac{\sum_{a\in\mathcal X_{\rm H}}|x_{\rm H}(a,q)|^2\omega_{\rm H}(m,q,a)}
{M_{\rm H}\psi_{\rm H}(m,q)},
\end{align}
where $\omega_{\rm H}(m,q,a)=\mathsf p_m(x_{\rm H}(a,q))$.
For an ordinary low-power subcarrier, we have
\begin{align}
\label{eq:rapamp_low_symbol_moments}
\mu_{\rm L}(m)
&=\frac{1}{M_{\rm L}\psi_{\rm L}(m)}
\sum_{a\in\mathcal X_{\rm L}}x_{\rm L}(a)\mathsf p_m(x_{\rm L}(a)), \\
\nu_{\rm L}(m)
&=\frac{1}{M_{\rm L}\psi_{\rm L}(m)}
\sum_{a\in\mathcal X_{\rm L}}|x_{\rm L}(a)|^2\mathsf p_m(x_{\rm L}(a)).
\end{align}
For an SI subcarrier carrying the local SI label vector $\boldsymbol\tau$, let $\mu_{\rm SI}(\boldsymbol\tau)=x_{\rm SI}(\boldsymbol\tau)$ and $\nu_{\rm SI}(\boldsymbol\tau)=|x_{\rm SI}(\boldsymbol\tau)|^2$. Given a candidate $\ell$, the local posterior moments are obtained block-wise. For $m\in\{m_{g,\ell,k}^H:k\in\mathcal I_r\}\cup\{m_{g,\ell,r}^{\rm SI}\}$, define
\begin{equation}
\label{eq:rapamp_cond_local_moment_kernel}
\begin{aligned}
(\mu_m(\boldsymbol\tau),\nu_m(\boldsymbol\tau))\triangleq
\begin{cases}
(\mu_{\rm H}(m,\tau_k),\nu_{\rm H}(m,\tau_k)), & m=m_{g,\ell,k}^H,\\
(\mu_{\rm SI}(\boldsymbol\tau),\nu_{\rm SI}(\boldsymbol\tau)), & m=m_{g,\ell,r}^{\rm SI}.
\end{cases}\nonumber
\end{aligned}
\end{equation}
Then the local posterior moments can be derived as
\begin{align}
\label{eq:rapamp_cond_block_moment}
\bar\mu_m^{(t)}(\ell)
&=\sum_{\boldsymbol\tau\in\{0,1\}^{|\mathcal I_r|}}\pi_{g,\ell,r}^{(t)}(\boldsymbol\tau)\mu_m(\boldsymbol\tau), \\
\bar\nu_m^{(t)}(\ell)
&=\sum_{\boldsymbol\tau\in\{0,1\}^{|\mathcal I_r|}}\pi_{g,\ell,r}^{(t)}(\boldsymbol\tau)\nu_m(\boldsymbol\tau).
\end{align}
For $m\in\mathcal D_{g,\ell}$, $\bar\mu_m^{(t)}(\ell)=\mu_{\rm L}(m)$ and $\bar\nu_m^{(t)}(\ell)=\nu_{\rm L}(m)$. Upon averaging these conditional moments over the pattern posterior, we can derive the denoiser output and its average posterior variance:
\begin{align}
\label{eq:rapamp_posterior_variance}
x_{{\rm den},m}^{(t)}
&=
\left[\mathcal D_{\rm PLIM}^{\rm F}
(\mathbf z_{\rm use}^{(t)},v_{\rm lin}^{(t)})\right]_m \nonumber \\
&=
\sum_{\ell\in\mathcal B_{\rm use}}
\rho_g^{(t)}(\ell)\bar\mu_m^{(t)}(\ell),
\quad m\in\mathcal G_g, \\
e_m^{(t)}
&=
\sum_{\ell\in\mathcal B_{\rm use}}
\rho_g^{(t)}(\ell)\bar\nu_m^{(t)}(\ell),
\quad m\in\mathcal G_g, \\
\bar v_{\rm post}^{(t)}&=
\frac{1}{N}\sum_{m=0}^{N-1}
\left(e_m^{(t)}-|x_{{\rm den},m}^{(t)}|^2\right).
\end{align}

Consequently, RA-AMP scores the current estimate as
\begin{equation}
\label{eq:rapamp_evidence_score}
S^{(t)}=\|\mathbf y-\mathbf H_{\rm eff}\widehat{\mathbf x}^{(t+1)}\|^2
+N_0\bar v_{\rm post}^{(t)}.
\end{equation}
If $S^{(t)}$ improves the best historical score, RA-AMP stores $\mathbf z_{\rm best}=\mathbf z_{\rm use}^{(t)}$ and $\mathbf v_{\rm best}=v_{\rm lin}^{(t)}\mathbf 1_N$. The iterations stop when
\begin{equation}
\label{eq:rapamp_stopping_delta}
\Delta^{(t)}=\frac{\|\widehat{\mathbf x}^{(t+1)}-
\widehat{\mathbf x}^{(t)}\|^2}{\|\widehat{\mathbf x}^{(t)}\|^2+\epsilon}
\end{equation}
falls below a tolerance threshold. The final pseudo-observation can be obtained as $\mathbf z_{\rm RA}=\mathbf z_{\rm best}$ with $\bar v_{\rm RA}=N^{-1}\sum_m v_{{\rm best},m}$.

RA-AMP then performs hard group detection using the stored pseudo-observation. For each pattern $\ell$, define the factorized log-MAP metric as
\begingroup
\setlength{\jot}{-1pt}
\begin{equation}
\begin{aligned}
\label{eq:rapamp_factorized_map_metric}
\Lambda_g(\ell)
=&
\ln p(\ell)
+\sum_{d\in\mathcal D_{g,\ell}}\ln\psi_{\rm L}(d) \\
&+\sum_{r=1}^{R_{\rm SI}}
\max_{\boldsymbol\tau\in\{0,1\}^{|\mathcal I_r|}}
\ln\varphi_{g,\ell,r}(\boldsymbol\tau),
\end{aligned}
\end{equation}
\endgroup
where all likelihoods are evaluated at $\mathbf z_{\rm RA}$ and $\bar v_{\rm RA}$. The MAP pattern and phase-bit estimates are
\begin{equation}
\label{eq:rapamp_final_pattern_decision}
\widehat\ell_g=
\argmax_{\ell\in\mathcal B_{\rm use}}\Lambda_g(\ell),
\end{equation}
\begin{equation}
\label{eq:rapamp_final_phase_decision}
\widehat{\mathbf q}_{g,\mathcal I_r}=
\argmax_{\boldsymbol\tau\in\{0,1\}^{|\mathcal I_r|}}
\varphi_{g,\widehat\ell_g,r}(\boldsymbol\tau),
\quad r=1,\ldots,R_{\rm SI}.
\end{equation}
Equations \eqref{eq:rapamp_factorized_map_metric}--\eqref{eq:rapamp_final_phase_decision} are equivalent to maximizing $P_g(\ell,\mathbf q|\mathbf z_{{\rm RA},g};\bar v_{\rm RA})$ over the complete PLIM-phase state space, but only require local SI-block searches. Once $\widehat\ell_g$ and $\widehat{\mathbf q}_g$ are obtained, the receiver demaps the index bits and phase bits. It then slices the high-power data symbols after compensating the detected phase rotations and slices the ordinary low-power data symbols with the unrotated low-power constellation. The proposed RA-AMP receiver is summarized in \textbf{Algorithm~\ref{alg:rapamp_receiver}}.

\vspace{-4mm}

\begin{algorithm}[t]
\caption{RA-AMP Receiver}
\label{alg:rapamp_receiver}
\footnotesize
\KwIn{$\mathbf y$, $\mathbf H_{\rm eff}$, $N_0$, $\mathcal B_{\rm use}$, $\mathcal X_{\rm H}$, $\mathcal X_{\rm L}$, $\alpha_{\rm H}$, $\alpha_{\rm L}$, $\theta_{\rm p}$, $\eta_{\rm RA}$, $\lambda_d$, $c_r$, $\epsilon$ and $T_{\max}$.}
\KwOut{Index bits, phase bits, and data symbol bits.}

Use $\mathbf H_{\rm eff}$ as the composite DAFT-domain channel matrix\;
Compute $\mathbf z_{\rm Eq}$, $\{\sigma_m^2\}$, and $v^{(0)}$ using \eqref{eq:rapamp_lmmse_filter}--\eqref{eq:rapamp_lmmse_variance}\;
Initialize $\widehat{\mathbf x}^{(0)}\leftarrow\mathcal D_{\rm PLIM}^{\rm F}(\mathbf z_{\rm Eq},v^{(0)})$, $\mathbf z_{\rm use}^{(-1)}\leftarrow\mathbf z_{\rm Eq}$, and $S_{\rm best}\leftarrow\infty$\;

\For{$t=0,1,\ldots,T_{\max}-1$}{
Compute $\mathbf r^{(t)}$ and $\mathbf z_{\rm lin}^{(t)}$ using \eqref{eq:rapamp_residual} and \eqref{eq:rapamp_linear_observation}\;
Estimate and clip $v_{\rm lin}^{(t)}$ using \eqref{eq:rapamp_variance_raw} and \eqref{eq:rapamp_variance_clip}\;
Form $\mathbf z_{\rm use}^{(t)}$ and run the factorized PLIM group-MMSE denoiser using \eqref{eq:rapamp_local_factor}--\eqref{eq:rapamp_posterior_variance}\;
Update $\widehat{\mathbf x}^{(t+1)}$ by the output damping in \eqref{eq:rapamp_memory_damping}\;
Compute $S^{(t)}$ using \eqref{eq:rapamp_evidence_score}; if improved, store $\mathbf z_{\rm best}=\mathbf z_{\rm use}^{(t)}$ and $\mathbf v_{\rm best}=v_{\rm lin}^{(t)}\mathbf 1_N$\;
Stop if $\Delta^{(t)}$ in \eqref{eq:rapamp_stopping_delta} is below the tolerance\;
}

Set $\mathbf z_{\rm RA}\leftarrow\mathbf z_{\rm best}$ and $\bar v_{\rm RA}\leftarrow N^{-1}\sum_m v_{{\rm best},m}$\;
\For{$g=0,1,\ldots,G-1$}{
Solve the factorized group MAP problem using \eqref{eq:rapamp_factorized_map_metric}--\eqref{eq:rapamp_final_phase_decision}\;
Demap $\widehat\ell_g$ and $\widehat{\mathbf q}_g$, slice the high-power and ordinary low-power data symbols, and output the corresponding bits\;
}
\end{algorithm}

\subsection{Complexity Analysis}

Let $T_{\rm max}$ denote the maximum number of iterations and $S$ represent the number of nonzero entries of $\mathbf H_{\rm eff}$. With scalar likelihood and moment tables precomputed for the current pseudo-observation, the factorized PLIM denoiser requires only local SI-block summations. Let us define the per-pattern factorized group-evaluation size as
\begin{equation}
\label{eq:complexity_rapamp_omega}
\Omega_{\rm F} \triangleq U+
\sum_{r=1}^{R_{\rm SI}}2^{|\mathcal I_r|}
\le U+R_{\rm SI}2^{B_{\rm L}}.
\end{equation}
Then, excluding the MMSE warm start, the overall complexity order of RA-AMP detection is given by
\begin{equation}
\label{eq:complexity_rapamp}
\begin{split}
&C_{\rm RA\text{-}AMP}
=\mathcal O\bigg(
\left(T_{\rm max}+2\right)GL\Omega_{\rm F} \\
&+\left(T_{\rm max}+2\right)N\left(M_{\rm H}+M_{\rm L}\right)
+T_{\rm max}\left(2S+N\right)
\bigg).
\end{split}
\end{equation}
The first term in \eqref{eq:complexity_rapamp} accounts for the initial denoising step, $T_{\rm max}$ iterative denoising steps, and the final factorized group MAP evaluation. The second term accounts for scalar likelihood and moment table construction, the third term corresponds to the two matrix-vector products $\mathbf H_{\rm eff}\widehat{\mathbf x}^{(t)}$ and $\mathbf H_{\rm eff}^{H}\mathbf r^{(t)}$ in each iteration, and the remaining term accounts for scalar damping operations. Compared with the full group denoiser, the factorization replaces the per-pattern phase-vector enumeration $2^{K_{\rm H}}$ by $\sum_{r=1}^{R_{\rm SI}}2^{|\mathcal I_r|}\le R_{\rm SI}2^{B_{\rm L}}$ without changing the posterior means or MAP decisions. Therefore, the final MAP stage can be reduced to a complexity of $\mathcal O(GL\Omega_{\rm F})$. Since $\mathbf H_{\rm eff}$ is sparse with $S=Nd$ and $U$, $B_{\rm L}$, $M_{\rm H}$, $M_{\rm L}$, and $T_{\rm max}$ are fixed, RA-AMP scales linearly with $N$. If the MMSE warm start is formed by a dense matrix inverse, an additional one-time $\mathcal O(N^3)$ complexity term should be included.

\section{Numerical Results}

This section evaluates the communication and sensing performance of the proposed PLIM-PR-AFDM scheme.
\vspace{-4mm}

\addtocounter{figure}{1}
\begin{figure*}[!t]
\centering
\begin{subfigure}[t]{0.325\textwidth}
\centering
\includegraphics[width=\linewidth]{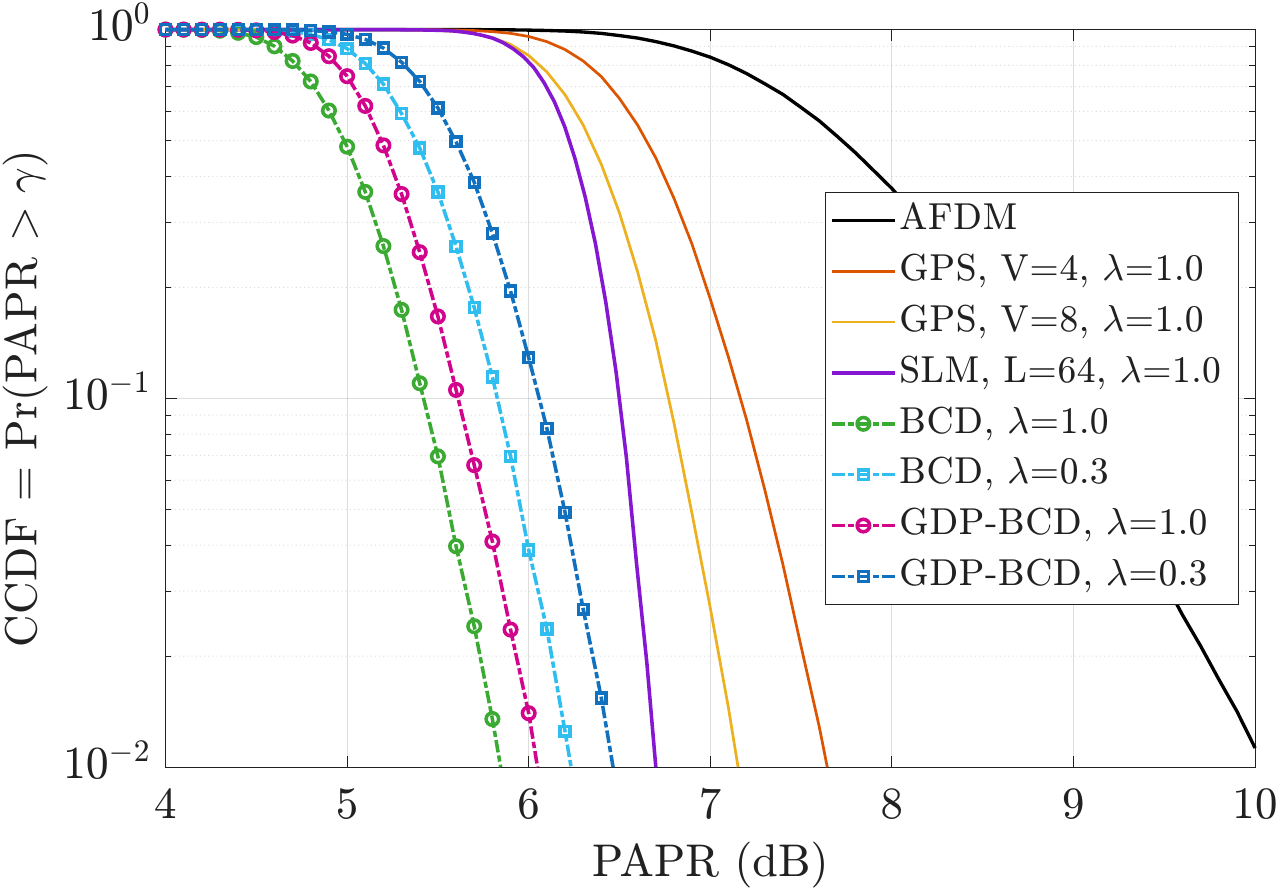}
\caption{PAPR CCDF.}
\label{fig:wave_opt_papr_ccdf}
\end{subfigure}
\hfill
\begin{subfigure}[t]{0.325\textwidth}
\centering
\includegraphics[width=\linewidth]{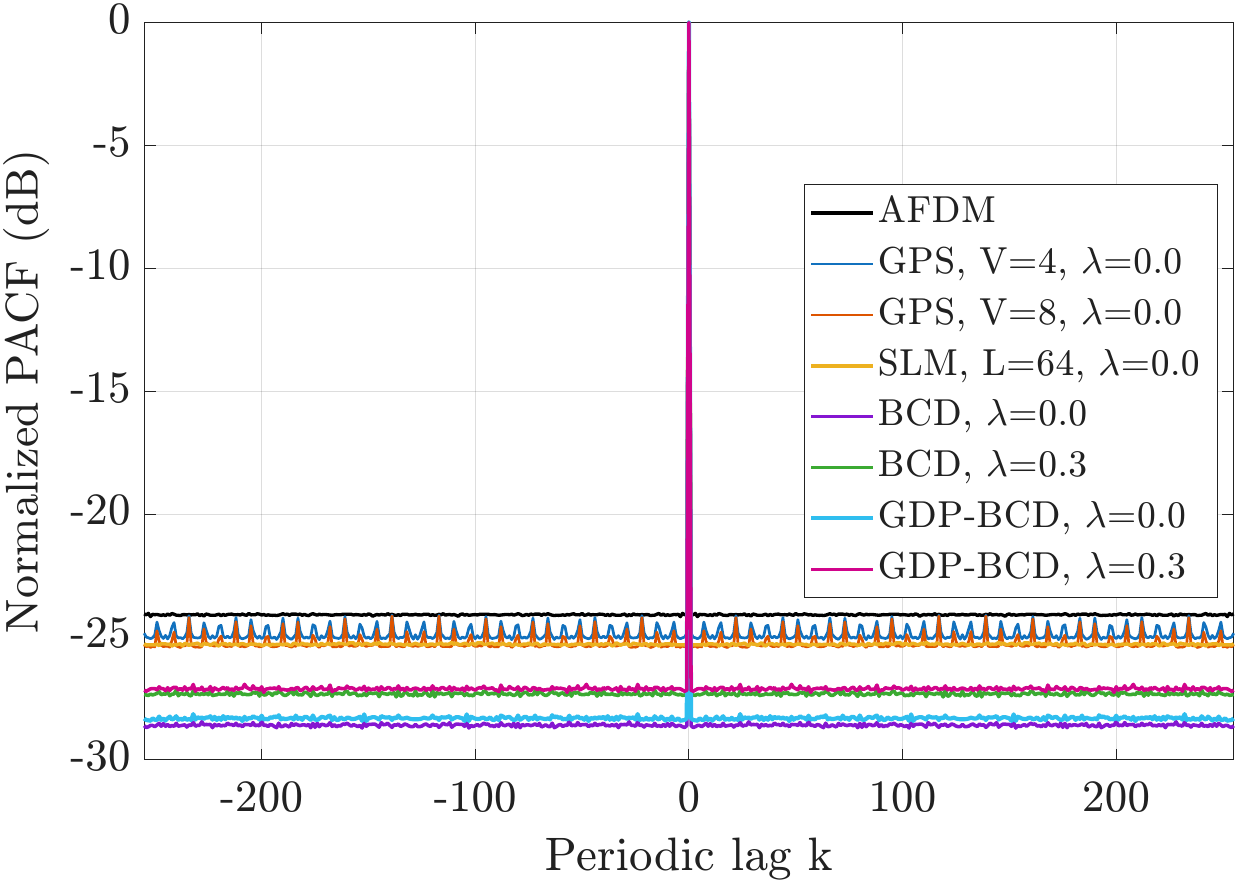}
\caption{Normalized PACF.}
\label{fig:wave_opt_pacf}
\end{subfigure}
\hfill
\begin{subfigure}[t]{0.325\textwidth}
\centering
\includegraphics[width=\linewidth]{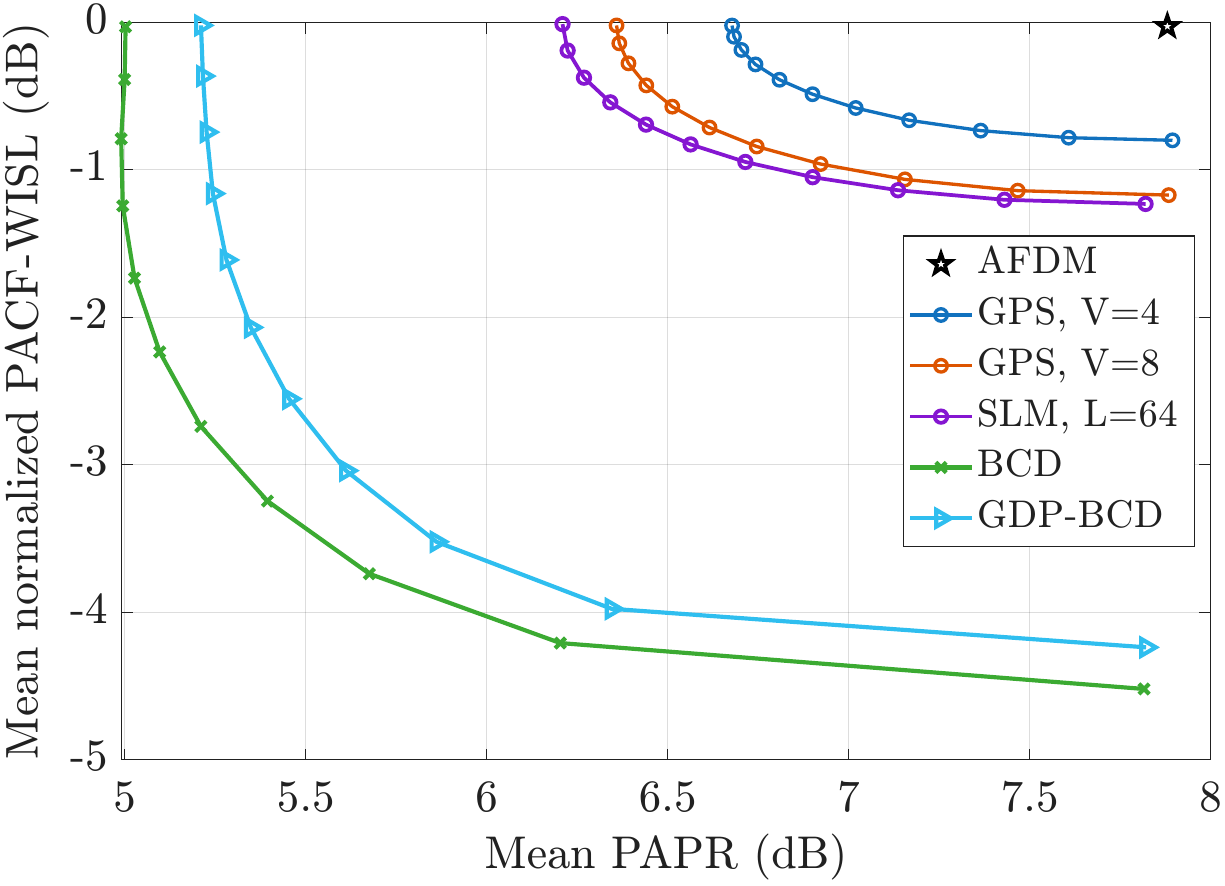}
\caption{Pareto fronts.}
\label{fig:wave_opt_pareto}
\end{subfigure}
\captionsetup{font=small}
\caption{Waveform optimization performance comparison for $N=256$, including PAPR reduction, zero Doppler cut sidelobe suppression, and the performance tradeoff of the proposed BCD/GDP-BCD algorithms against GPS, SLM, and conventional AFDM baselines.}
\label{fig:wave_opt_comparison}
\vspace{-6mm}
\end{figure*}

\addtocounter{figure}{-2}
\begin{figure}[!t]
\centering
\includegraphics[width=0.7\linewidth]{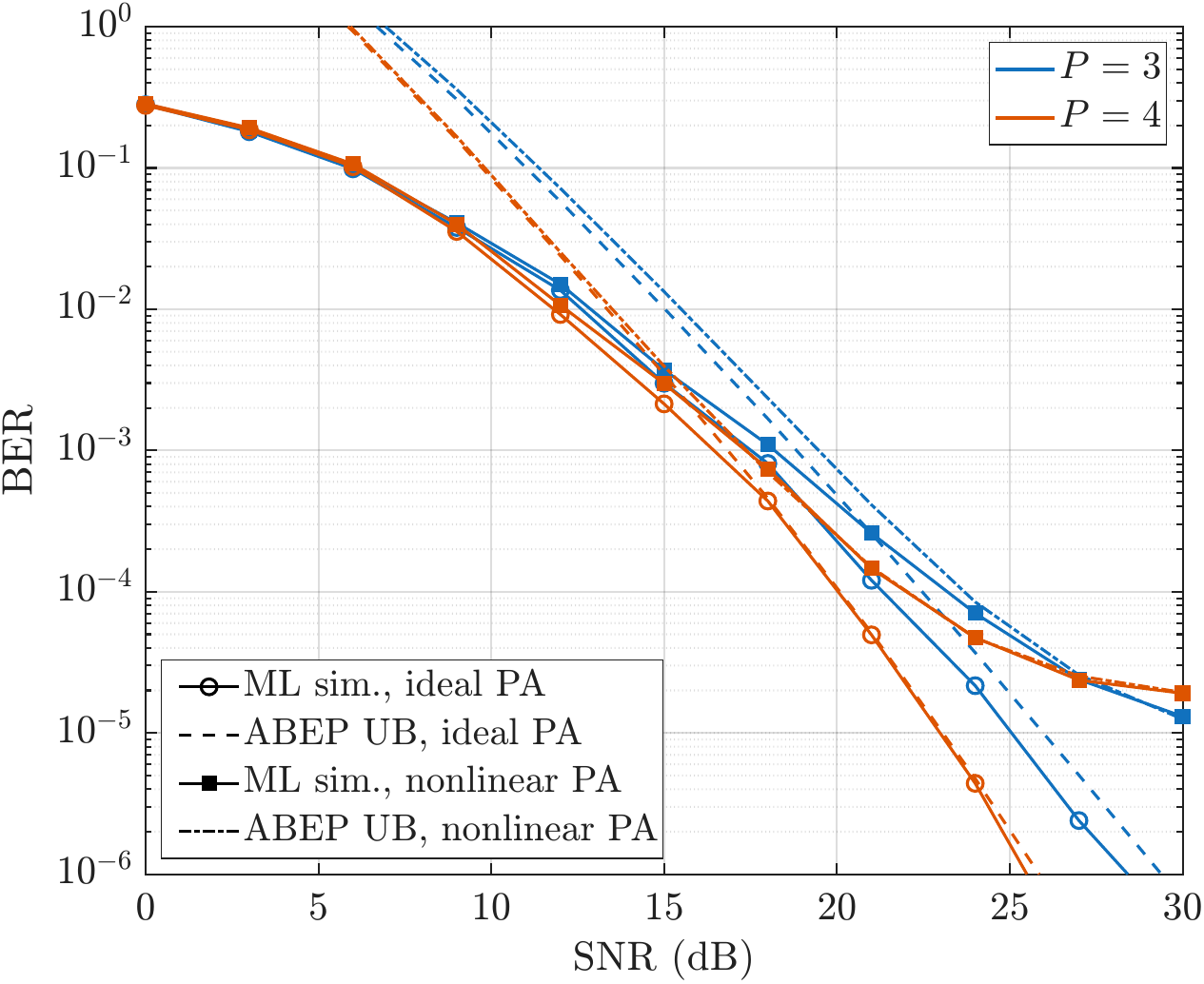}
\caption{Comparison between the theoretical upper bound on the ABEP and simulated BER for the proposed system under varying numbers of paths.}
\label{fig:ABEP_N6_U3}
\vspace{-6mm}
\end{figure}
\addtocounter{figure}{1}

\subsection{Theoretical ABEP Analysis}
\label{subsec:theoretical_abep}
Fig.~\ref{fig:ABEP_N6_U3} compares the theoretical upper bound on the ABEP with the simulated ML BER under ideal and nonlinear PA cases. The parameters are set to $N=6$, $U=3$, $K_{\rm H}=1$, $\theta_{\rm p}=\pi/5$, $(M_{\rm H},M_{\rm L})=(4,2)$, $\rho=2.8$, $c_1=5/(2N)$, and $c_2=\sqrt{2}$, with IBO $=3$~dB for the nonlinear PA, where the nonlinear PA is characterized by the modified Rapp model~\cite{nokia2016realistic}. For $P=3$, the delay and normalized Doppler vectors are $[0,1,3]$ and $[0.2,0.5,0.7]$, respectively, and for $P=4$, they are $[0,1,3,5]$ and $[0.2,0.4,0.55,0.7]$, respectively. The channel gains follow $\mathbf h\sim\mathcal{CN}(\mathbf 0,\mathbf I_P/P)$. The theoretical ABEP bounds closely follow the simulated BER in the medium-to-high SNR regions, validating the proposed error analysis. Increasing the number of paths from $P=3$ to $P=4$ increases the ideal-PA BER slope, which is consistent with the expected multipath-diversity gain. With nonlinear PA distortion, however, we observe an error floor around $10^{-5}$, which weakens the additional gain from larger $P$ and emphasizes the need for PAPR-aware waveform design.
\vspace{-6mm}
\subsection{Evaluation of the Proposed Optimization Algorithm}
\label{subsec:wav_opt}
This subsection evaluates the performance of the proposed waveform optimization method. The AFDM block length is set to $N=256$, with chirp parameters $c_1=35/(2N)$ and $c_2=\sqrt{2}$. The benchmarks are the GPS algorithm~\cite{AFDM_GPS}, with $W=4$ candidates, and the selected mapping (SLM) scheme~\cite{PAPR_survey}, with $L=64$ phase sequences. Both are adapted to the same multi-objective model. The proposed PLIM-aided 8PSK-QPSK system harnesses $\{U, K_{\rm H}, M_{\rm H}, M_{\rm L}, \rho, \theta_{\rm p}\} = \{4, 2, 8, 4, 4, \pi/5\}$. Both BCD and GDP-BCD use $T=5$ and $\epsilon=10^{-10}$. Furthermore, for GDP-BCD, the group and candidate-pattern budgets are $K_G=2$ and $K_C=\lceil 0.35N_{\rm SI} \rceil = 23$, respectively.

Fig.~\ref{fig:wave_opt_comparison} compares the waveform optimization performance of the proposed phase-rotation design with several AFDM benchmarks. As shown by the PAPR complementary cumulative distribution function (CCDF) in Fig.~\ref{fig:wave_opt_comparison}(\subref{fig:wave_opt_papr_ccdf}), the conventional AFDM waveform achieves the highest PAPR, while the GPS and SLM schemes provide moderate PAPR reduction. The proposed BCD and GDP-BCD designs further reduce the PAPR. Specifically, $\lambda=1$ yields the largest PAPR reduction, whereas $\lambda=0.3$ trades a slight PAPR penalty for improved correlation sidelobe behavior. Fig.~\ref{fig:wave_opt_comparison}(\subref{fig:wave_opt_pacf}) further shows that the conventional AFDM waveform has high PACF sidelobes because its random data symbols are not optimized for the instantaneous correlation structure, while the suppression achieved by the GPS and SLM schemes is limited by their small codebook spaces. The proposed algorithms suppress most nonzero-lag sidelobes more effectively, with $\lambda=0$ yielding the strongest PACF suppression and $\lambda=0.3$ maintaining a low sidelobe floor while providing a meaningful PAPR reduction. This stronger PACF suppression suggests that the PAPR-reduction algorithms converge earlier than PACF-WISL minimization, while the nearly identical BCD and GDP-BCD results confirm that gradient-guided pruning preserves most of the waveform-shaping capability of exhaustive block evaluation.

Fig.~\ref{fig:wave_opt_comparison}(\subref{fig:wave_opt_pareto}) illustrates the tradeoff between the mean PAPR and the mean normalized PACF-WISL. The conventional AFDM point remains in the high-PAPR/high-WISL region, while both the GPS and SLM schemes exhibit relatively poor optimization performance. By varying $\lambda$, the proposed algorithms form a lower-left Pareto frontier, showing that the embedded phase-code DoF can balance PA efficiency and sensing sidelobe suppression. Moreover, GDP-BCD nearly overlaps with BCD on the Pareto front, indicating that most of the gains are retained while substantially fewer exact objective evaluations are required. Although the first-order gradient approximation in GDP-BCD makes it slightly less optimal and more prone to local optima, Figs.~\ref{fig:wave_opt_comparison}(\subref{fig:wave_opt_papr_ccdf}) and~\ref{fig:wave_opt_comparison}(\subref{fig:wave_opt_pacf}) show that the main PAPR and PACF sidelobes are suppressed in the early iterations, with later updates bringing only marginal refinements. Therefore, the performance loss is negligible.

\begin{figure}[!t]
\centering
\centering
\includegraphics[width=0.68\linewidth]{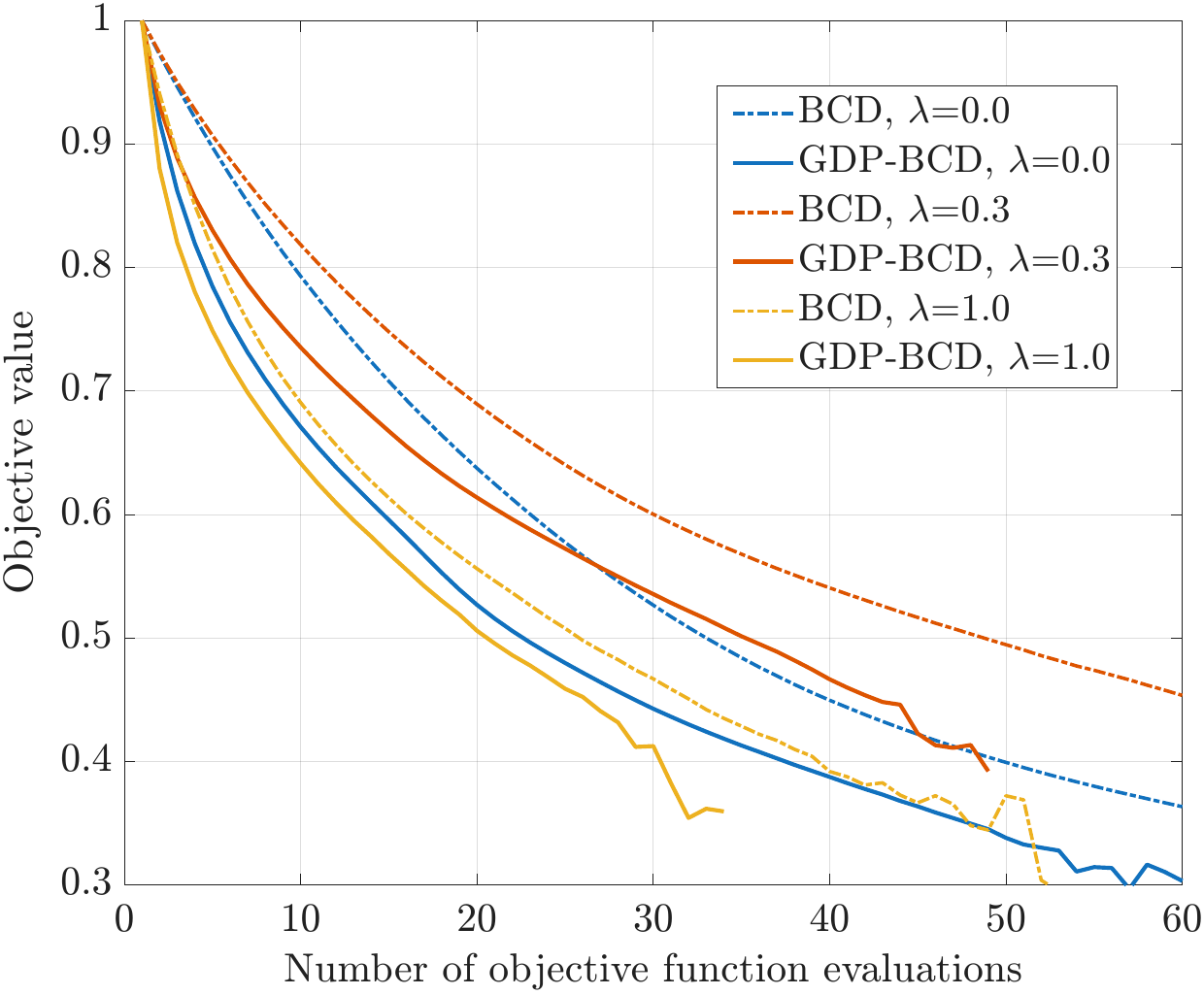}
\caption{Convergence trajectories of the objective function for BCD and GDP-BCD under different waveform-optimization weights $\lambda$.}
\captionsetup{font=small}
\label{fig:wave_opt_convergence}
\vspace{-2mm}
\end{figure}

\begin{table}[!t]
  \renewcommand{\arraystretch}{1.3}
  \captionsetup{font=footnotesize,labelfont=normalfont,textfont=sc,
    labelsep=newline,justification=centering,singlelinecheck=false}
  \caption{Average Runtime Comparison of BCD and GDP-BCD}
  \label{tab:wave_opt_runtime}
  \centering
  \footnotesize
  \begin{tabular}{|c|c|}
    \hline
    Algorithm & Average runtime (s) \\
    \hline
    BCD & 0.0545 \\
    \hline
    GDP-BCD & 0.0200 \\
    \hline
  \end{tabular}
  \vspace{-6mm}
\end{table}

Fig.~\ref{fig:wave_opt_convergence} further evaluates the convergence behavior of the proposed optimization algorithms. As shown in Fig.~\ref{fig:wave_opt_convergence}, the objective values of both BCD and GDP-BCD decrease rapidly in the initial iterations and then enter a slower refinement stage, indicating that the dominant phase-code updates can be identified within a limited number of iterations. For all tested values of $\lambda$, GDP-BCD follows the BCD descent trend and achieves comparable objective values, thereby verifying that the gradient-guided pruning strategy preserves the main optimization trajectory while avoiding unnecessary exact objective evaluations. The mixed-objective case $\lambda=0.3$ converges more slowly than the single-objective cases since the algorithm may jointly balance PAPR reduction and PACF-WISL suppression. Table~\ref{tab:wave_opt_runtime} further reports the average runtime. Compared with BCD, GDP-BCD reduces the average runtime from $0.0545$~s to $0.0200$~s, corresponding to a $63.3\%$ reduction. This substantial saving confirms the complexity advantage of the dual-pruned candidate screening mechanism.

\begin{figure}[!t]
\centering
\begin{subfigure}[t]{0.75\linewidth}
\centering
\includegraphics[width=\linewidth]{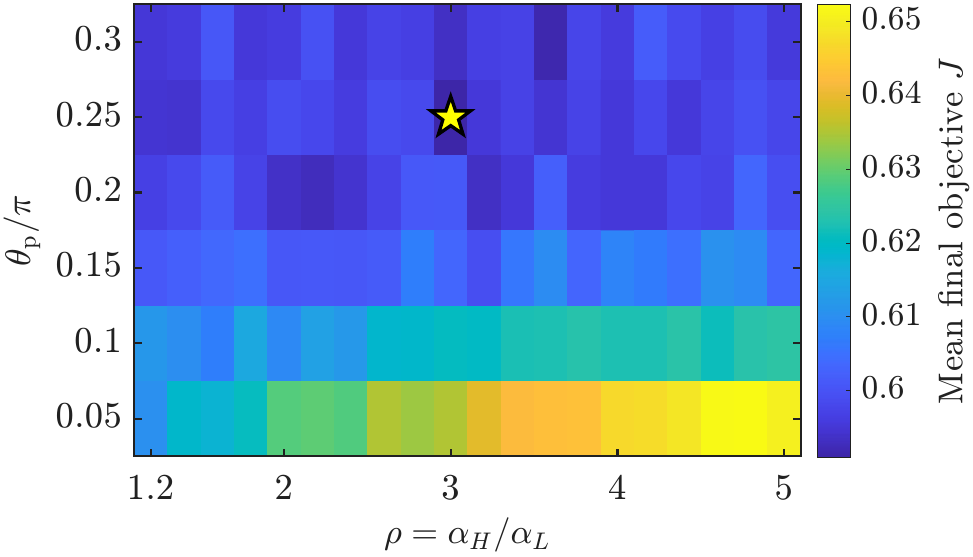}
\caption{Waveform optimization objective achieved by GDP-BCD over the power ratio $\rho$ and phase-rotation amplitude $\theta_{\rm p}$ with $\lambda=0.5$.}
\label{fig:objective_rho_sweep}
\end{subfigure}
\begin{subfigure}[t]{0.75\linewidth}
\centering
\includegraphics[width=\linewidth]{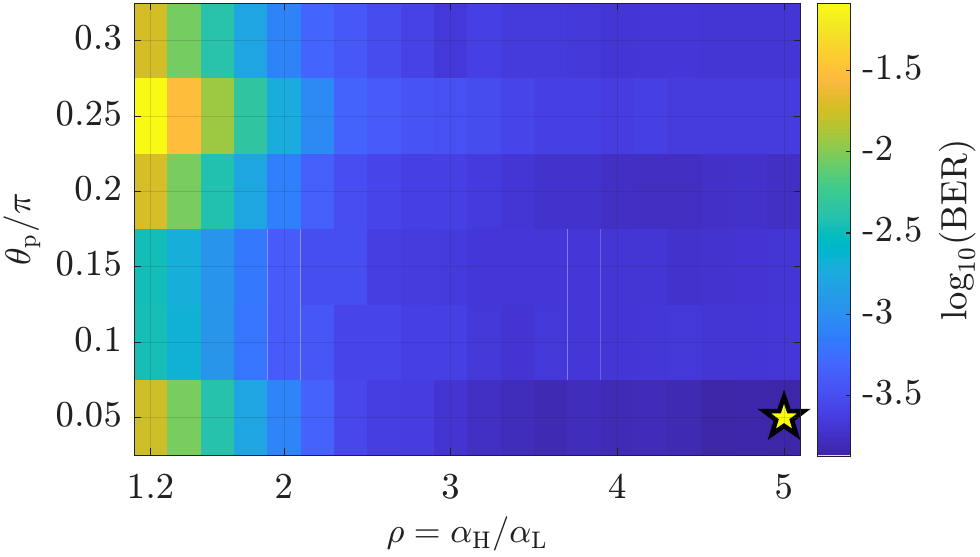}
\caption{BER over the power ratio $\rho$ and phase-rotation amplitude $\theta_{\rm p}$ in doubly dispersive channel at $E_b/N_0=15$~dB.}
\label{fig:ber_rho_sweep}
\end{subfigure}
\caption{Joint parameter sweep of the PLIM power ratio $\rho$ and phase-rotation amplitude $\theta_{\rm p}$, showing their effects on the optimized waveform objective and communication reliability.}
\captionsetup{font=small}
\label{fig:results_rho_theta_sweep}
\vspace{-6mm}
\end{figure}

\begin{figure}[!t]
\centering
\captionsetup[subfigure]{justification=centering,singlelinecheck=true}
\begin{subfigure}[t]{\linewidth}
\centering
\includegraphics[width=0.81\linewidth]{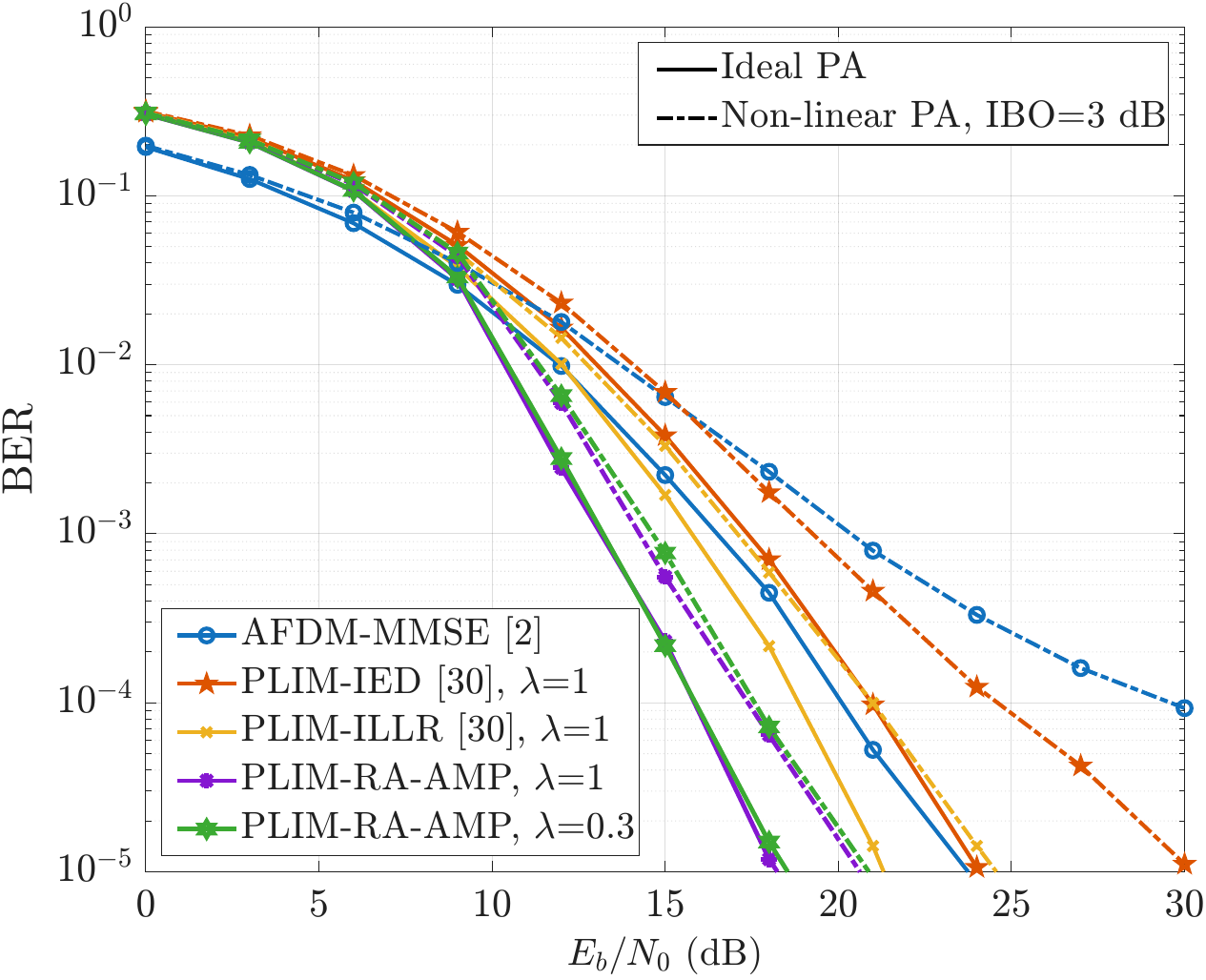}
\caption{Proposed 8PSK-QPSK configuration with $\rho=4$.}
\label{fig:ber_8psk_rho40_ibo}
\end{subfigure}
\vspace{0.5em}
\begin{subfigure}[t]{\linewidth}
\centering
\includegraphics[width=0.81\linewidth]{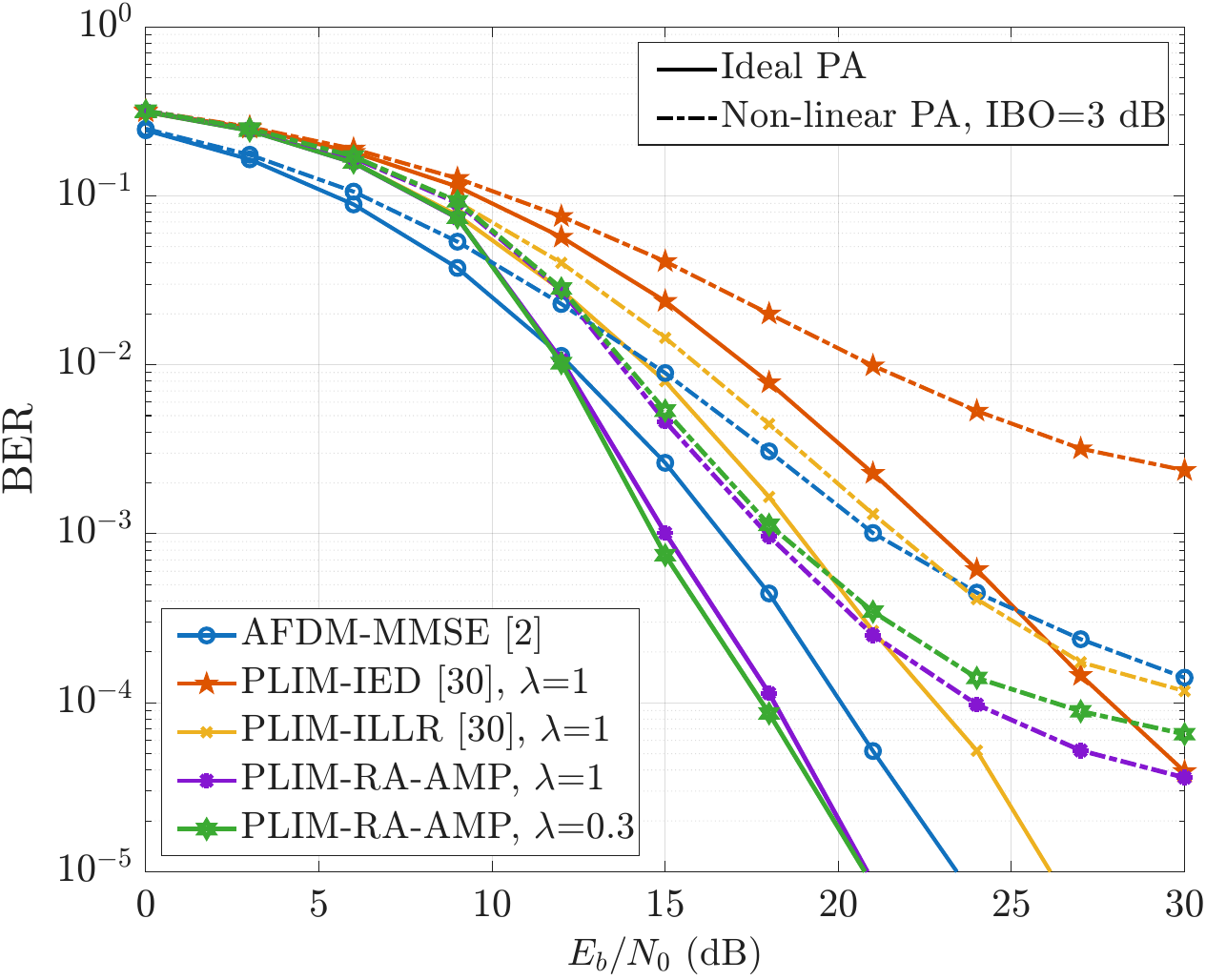}
\caption{Proposed 16QAM-QPSK configuration with $\rho=3$.}
\label{fig:ber_16qam_rho30_ibo}
\end{subfigure}
\caption{BER performance under doubly dispersive channels with ideal PA and nonlinear PA distortion at IBO $=3$ dB: (a) 8PSK-QPSK, and (b) 16QAM-QPSK.}
\captionsetup{font=small}
\label{fig:ber_nonlinear_pa}
\vspace{-6mm}
\end{figure}

Fig.~\ref{fig:results_rho_theta_sweep} illustrates the impact of the power ratio $\rho=\alpha_{\mathrm{H}}/\alpha_{\mathrm{L}}$ and the phase-rotation amplitude $\theta_{\rm p}$. Fig.~\ref{fig:results_rho_theta_sweep}(\subref{fig:objective_rho_sweep}) depicts the waveform-optimization objective with respect to $\rho$ and $\theta_{\rm p}$. It is observed that the optimization performance deteriorates when $\theta_{\rm p}$ is small, whereas it becomes relatively insensitive to variations in $\rho$ for large values of $\theta_{\rm p}$. Furthermore, Fig.~\ref{fig:results_rho_theta_sweep}(\subref{fig:ber_rho_sweep}) presents the distribution of the BER over doubly dispersive channel. Excessively small values of $\rho$ degrade the BER performance, whereas the BER remains resilient to $\theta_{\rm p}$. Therefore, a relatively large value of $\theta_{\rm p}$ and a moderately large value of $\rho$ are recommended for practical systems to achieve a favorable trade-off between waveform optimization effectiveness and communication reliability.
\vspace{-4mm}
\subsection{ISAC Performance Comparison}
This subsection further evaluates and validates the proposed waveform through simulations in a concrete ISAC scenario. The doubly dispersive channel is configured with $N=256$ subcarriers and $P=3$ propagation paths. The carrier frequency, subcarrier spacing, and relative velocity are set to $f_c=4$~GHz, $\Delta f=15$~kHz, and $150$~km/h, respectively. The baselines include conventional AFDM with MMSE equalization, as well as the improved energy-detection (IED) and improved log-likelihood-ratio (ILLR) detectors according to~\cite{DM_OFDM_CP} for the proposed PLIM-PR-AFDM system. To ensure a fair comparison, the conventional AFDM baseline is configured to achieve the same information rate as the proposed scheme.

Fig.~\ref{fig:ber_nonlinear_pa} compares the BER over doubly dispersive channels for the 8PSK-QPSK and 16QAM-QPSK configurations using an ideal PA and a nonlinear PA with the same configuration as that in Subsection~\ref{subsec:theoretical_abep}. Specifically, $\rho=4$ is used for the 8PSK-QPSK configuration, whereas $\rho=3$ is employed for the 16QAM-QPSK configuration. For the RA-AMP receiver, the algorithm parameters are set to $\eta_{\rm RA}=0.9$, $\lambda_d=0.65$, $T_{\max}=4$, and $c_r=0.20$, while the remaining waveform configurations are identical to those described in Subsection~\ref{subsec:wav_opt}.  As observed in Fig.~\ref{fig:ber_nonlinear_pa}(\subref{fig:ber_8psk_rho40_ibo}), with the 8PSK-QPSK configuration, the proposed scheme outperforms the baselines and conventional AFDM using both PA configurations. At a BER of $10^{-3}$, PLIM-PR-AFDM with RA-AMP achieves SNR gains of $4$~dB and $6$~dB using ideal and nonlinear PAs, respectively, over conventional AFDM, thereby demonstrating robustness against doubly dispersive channels and PA nonlinearity. Fig.~\ref{fig:ber_nonlinear_pa}(\subref{fig:ber_16qam_rho30_ibo}) presents the results using the 16QAM-QPSK setting, where the proposed scheme also maintains the best performance among all the scenarios. Although the performance gain over the baseline schemes decreases slightly compared with that in Fig.~\ref{fig:ber_nonlinear_pa}(\subref{fig:ber_8psk_rho40_ibo}), the proposed scheme still achieves SNR gains of approximately $1.5$~dB and $3$~dB at a BER of $10^{-3}$ using ideal and nonlinear PAs, respectively. Notably, the baselines, especially the ILLR detector, also achieve better BER performance than conventional AFDM under PA nonlinearity, which can be attributed to the low PAPR of the optimized waveform itself. By contrast, the IED detector exhibits degraded BER because it only exploits the high- and low-power information. Meanwhile, the proposed RA-AMP algorithm provides an additional performance gain of about $3$~dB over the ILLR detector. This advantage arises because the algorithm effectively mitigates the adverse effects of the non-constant-modulus properties through iterative denoising, thereby achieving superior BER performance.

\begin{figure}[!t]
\centering
\includegraphics[width=0.95\linewidth]{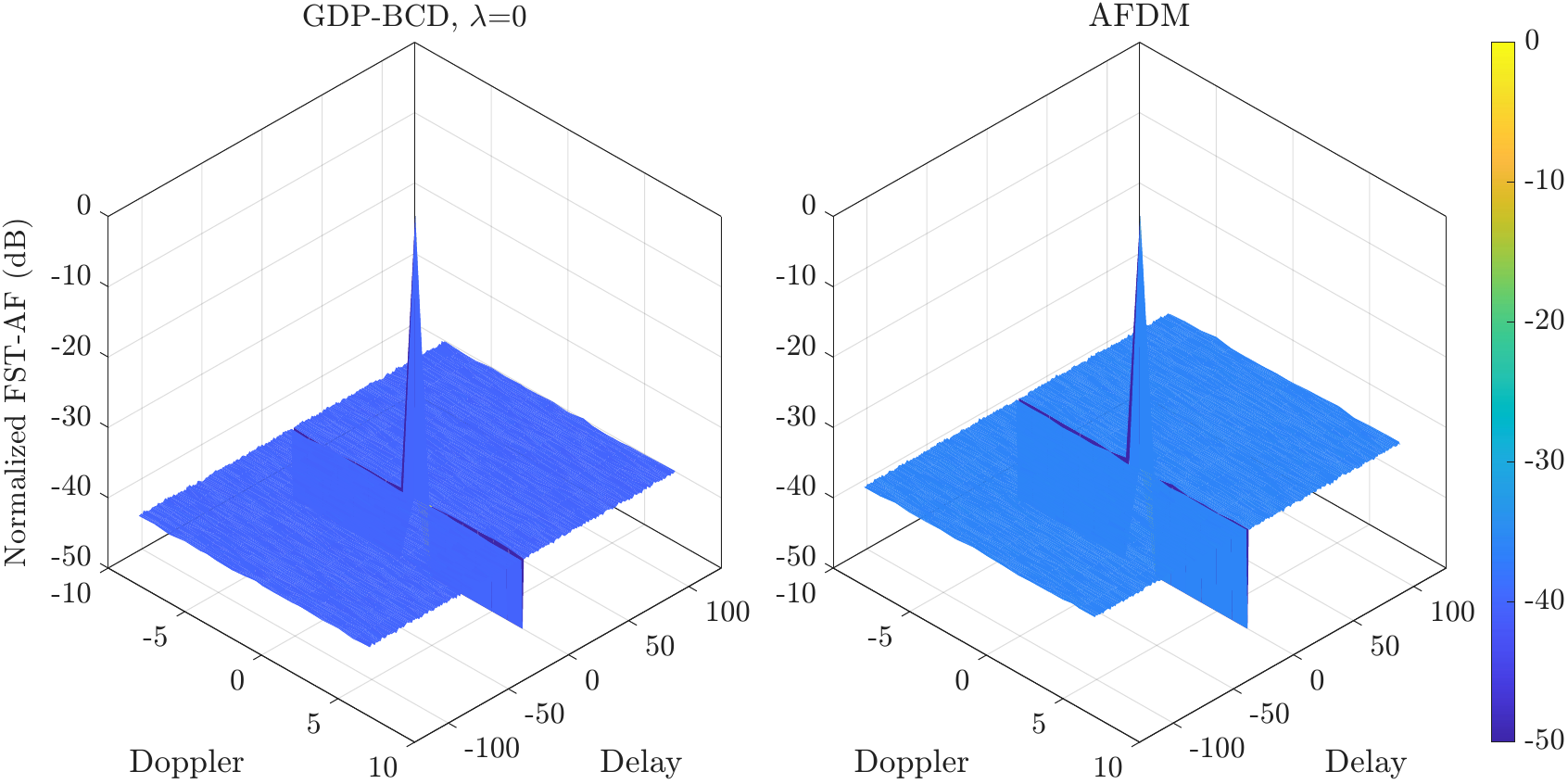}
\caption{Normalized FST-AF of the optimized 8PSK-QPSK AFDM waveform obtained by GDP-BCD with $\lambda=0$ and the conventional AFDM waveform, with $N=256$ and $M_{\rm s}=16$.}
\label{fig:fst_af_surface_8psk_n256_m16}
\vspace{-4mm}
\end{figure}

\begin{figure}[!t]
\centering
\captionsetup[subfigure]{justification=centering,singlelinecheck=true}
\begin{subfigure}[t]{\linewidth}
\centering
\includegraphics[width=0.75\linewidth]{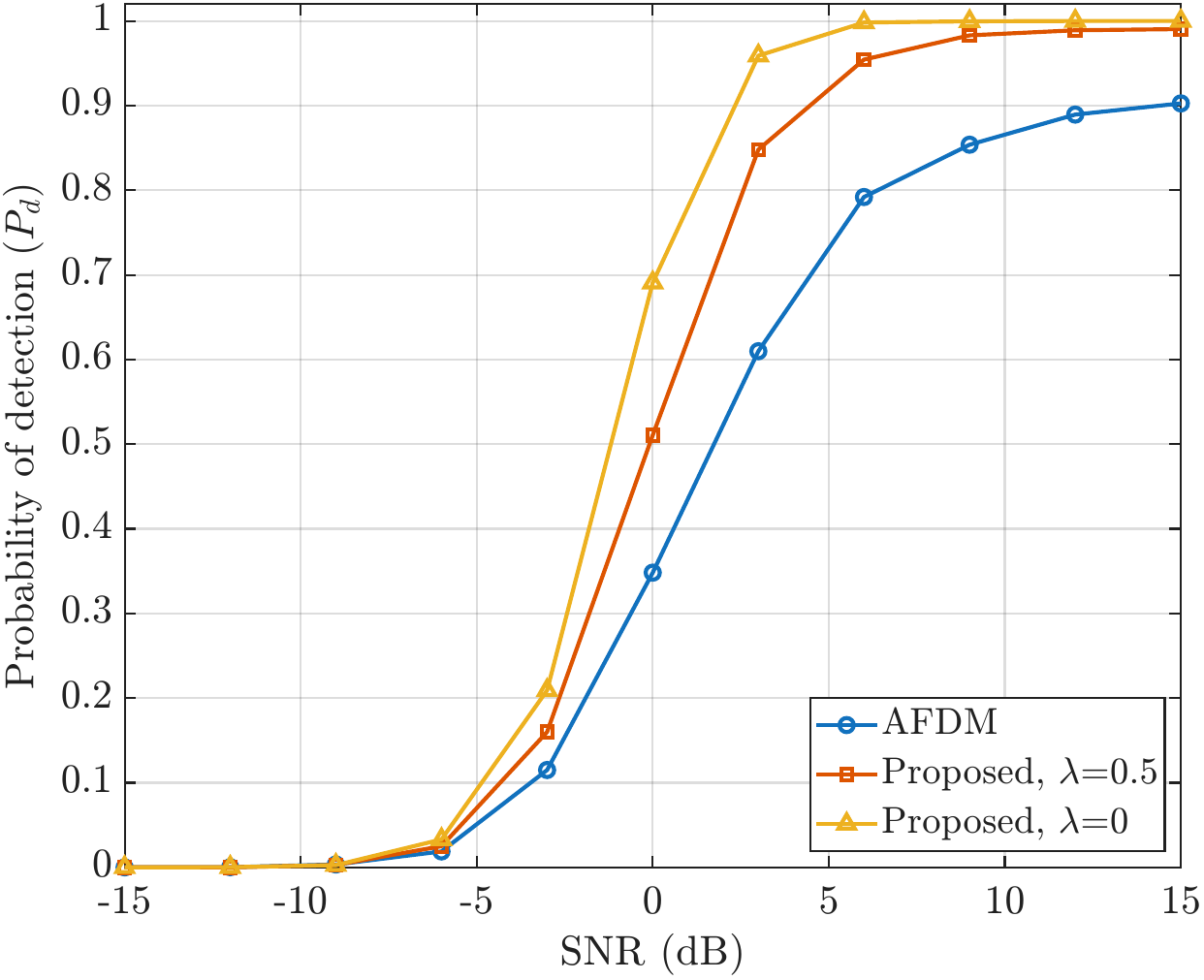}
\caption{CA-CFAR probability of detecting all targets versus SNR.}
\label{fig:sensing_pd_multi_target}
\end{subfigure}
\begin{subfigure}[t]{\linewidth}
\centering
\includegraphics[width=0.84\linewidth]{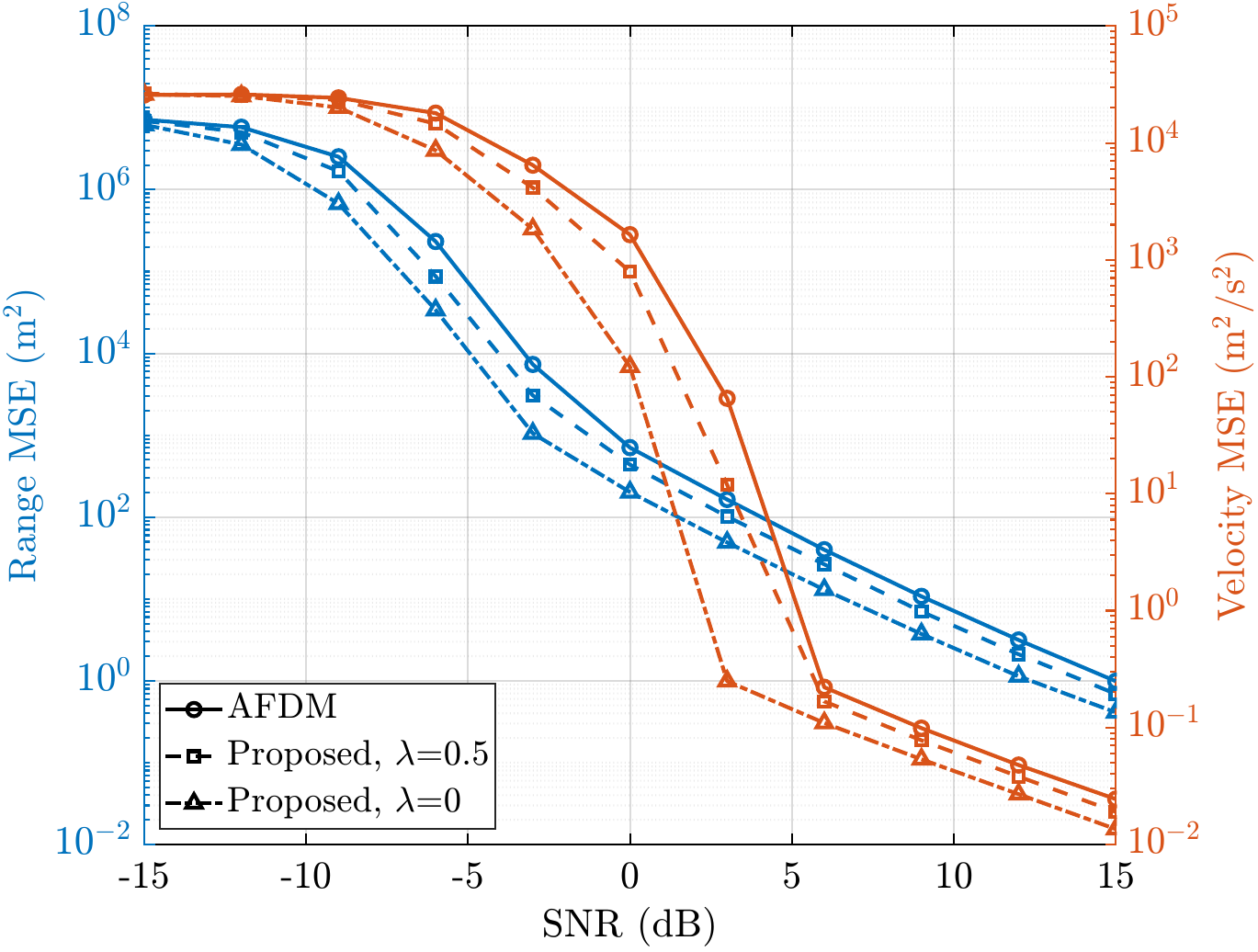}
\caption{TLS-ESPRIT range and velocity estimation MSE versus SNR.}
\label{fig:sensing_mse_tls_esprit}
\end{subfigure}
\caption{Sensing performance of the conventional AFDM waveform and the proposed optimized waveforms with $\lambda=0.5$ and $\lambda=0$: (a) probability of detecting all targets using CA-CFAR, and (b) range and velocity MSE using TLS-ESPRIT.}
\captionsetup{font=small}
\label{fig:sensing_mse_pd_multi_target}
\vspace{-6mm}
\end{figure}

Fig.~\ref{fig:fst_af_surface_8psk_n256_m16} compares the normalized FST-AFs of the sensing-oriented optimized waveform and the conventional AFDM waveform. Both waveforms preserve the mainlobe at $0$~dB, while the proposed GDP-BCD design suppresses the sidelobes. Specifically, the conventional AFDM waveform exhibits a broad sidelobe floor of about $-36$~dB, whereas the optimized waveform reduces most sidelobes to approximately $-40$~dB.

The sensing results in Fig.~\ref{fig:sensing_mse_pd_multi_target} further illustrate the benefits of the optimized waveform. The coherent processing interval contains $M_{\rm s}=16$ AFDM blocks, and the received echo is processed by a matched filter, followed by cell-averaging constant false alarm rate (CA-CFAR) detection and total least squares estimation of signal parameters through rotational invariance techniques (TLS-ESPRIT)-based range/velocity estimation~\cite{Signal_process}. The CA-CFAR false alarm rate is $10^{-6}$, with 4 training cells and 1 guard cell. For the considered case with $Q=3$ targets, the delay bins are set to $\{3,9,14\}$ and the Doppler bins to $\{3,-4,6\}$. The reflection coefficients are $\boldsymbol{\beta}=[1,10^{-1/2},10^{-1}]$. As shown in Fig.~\ref{fig:sensing_mse_pd_multi_target}(\subref{fig:sensing_pd_multi_target}), the proposed waveforms achieve about $3$--$6$~dB SNR gain in probability of detection $P_{\rm d}$ over conventional AFDM under different weighting factors. Moreover, Fig.~\ref{fig:sensing_mse_pd_multi_target}(\subref{fig:sensing_mse_tls_esprit}) presents the single-target TLS-ESPRIT parameter-estimation performance. It can be observed that the proposed method yields about $2$--$3$~dB SNR gain under different weights. These gains result from FST-AF sidelobe suppression, which alleviates inter-target interference and improves sensing performance.
\vspace{-3mm}
\section{Conclusion}
This paper proposed a PLIM-PR-AFDM framework for ISAC, where each PLIM power pattern conveys index bits. Explicitly, high-power subcarriers provide phase-code DoF for frame-wise waveform shaping, and selected low-power symbols embed the phase bits as SI without an external channel. By relating the FST-AF sidelobe objective to per-block PACF-WISL minimization, the proposed scheme jointly optimizes the PAPR and the sidelobe levels of the FST-AF through the proposed BCD/GDP-BCD algorithms. The BER upper-bound and full-diversity condition were derived for ideal and nonlinear-PA cases, and an efficient RA-AMP receiver was developed for joint detection of the PLIM power pattern, phase bits, and data symbol bits. Numerical results verified the correctness of the theoretical error analysis and demonstrated favorable PAPR-versus-sidelobe tradeoffs. Furthermore, our RA-AMP detector is capable of achieving good BER performance under doubly dispersive channels and PA nonlinearity, and PLIM-PR-AFDM may attain improved target detection and estimation accuracy.

\bibliographystyle{IEEEtran}
\bibliography{references, additional_references}

\end{document}